\documentclass[a4paper,10pt,pointlessnumbers,DIV11,abstracton]{scrartcl}

\usepackage{standard_packages}
\usepackage{diagxy}

\numberwithin{equation}{section}

\providecommand{\spec}{\mathrm{spec}}

\providecommand{\PAc}{{\mathsf{p}^A_{c}}}
\providecommand{\PiAc}{{\mathsf{P}^A_{c}}}
\providecommand{\Qe}{{Q_{\eps}}}
\providecommand{\PiAel}{{\Pi^{A}_{\eps,\lambda}}}

\providecommand{\Hd}{H_{\mathrm{D}}}

\providecommand{\PSpace}{T^{\ast} \R^d}

\providecommand{\ordere}[1]{\mathcal{O}(\eps^{#1})}
\providecommand{\orderl}[1]{\mathcal{O}(\lambda^{#1})}
\providecommand{\orderone}{\mathcal{O}(1)}
\providecommand{\eprec}{\epsilon}
\providecommand{\ordereprec}{\mathcal{O}(\eprec +)}

\providecommand{\orderc}[1]{\mathcal{O}(\nicefrac{1}{c^{#1}})}

\providecommand{\Piref}{\Pi_{\mathrm{ref}}}
\providecommand{\piref}{\pi_{\mathrm{ref}}}

\providecommand{\Hfast}{\mathcal{H}_{\mathrm{fast}}}
\providecommand{\Hslow}{\mathcal{H}_{\mathrm{slow}}}

\providecommand{\Fs}{\mathcal{F}_{\sigma}}

\providecommand{\Weyl}{\sharp}
\providecommand{\Weyle}{\Weyl_{\eps}}
\providecommand{\magB}{\sharp^B}

\providecommand{\magBl}{\Weyl^B_{\lambda}}
\providecommand{\magWel}{\Weyl^B_{\eps,\lambda}}
\providecommand{\magBc}{\sharp^B_c}

\providecommand{\GA}{{\Gamma^A}}
\providecommand{\GAe}{{\Gamma^A_{\eps}}}

\providecommand{\OBel}{{\Omega^B_{\eps,\lambda}}}
\providecommand{\GAe}{{\Gamma^A_{\eps}}}
\providecommand{\gBe}{{\gamma^B_{\eps}}}
\providecommand{\OBel}{\Omega^B_{\eps,\lambda}} 

\providecommand{\Ael}{A^{\eps,\lambda}}

\providecommand{\Bte}{\tilde{B}^{\eps}}

\providecommand{\OpAsr}{\mathrm{Op}^A_c} 
\providecommand{\Op}{\mathrm{Op}}
\providecommand{\OpAus}{\Op^A_{\mathrm{u}}}
\providecommand{\OpAel}{\mathrm{Op}^A_{\eps,\lambda}}
\providecommand{\Ope}{\mathrm{Op}_{\eps}}

\providecommand{\WeylSysAel}{{W^A_{\eps,\lambda}}}
\providecommand{\WeylUs}{W^A_{\mathrm{u}}}

\providecommand{\WignerTrafoel}{{\mathcal{W}^A_{\eps,\lambda}}}
\providecommand{\WignerUs}{\tilde{\mathcal{W}}^A}

\providecommand{\minsAl}{\vartheta^A_{\lambda}}

\providecommand{\Hoerrd}[1]{\mathcal{S}^{#1}_{\rho}}
\providecommand{\Hoermrd}[3]{\mathcal{S}^{#1}_{#2}}

\providecommand{\SemiHoerrd}[1]{\mathrm{A}\Hoerrd{#1}}
\providecommand{\SemiHoermrd}[3]{\mathrm{A}\Hoermrd{#1}{#2}{#3}}

\providecommand{\Schwartz}{\mathcal{S}}

\begin{document}

\date{\today}

\title{Two-parameter Asymptotics \\
in Magnetic Weyl Calculus}

\author{Max Lein}

\publishers{
\bigskip
\normalsize
Technische Universität München, Zentrum Mathematik, Boltzmannstraße 3, 85747 Garching, Germany 
\\
\smallskip
e-mail: lein@ma.tum.de}

\hypersetup{
	pdfauthor={Max Lein}, 
	pdftitle={Two-parameter Asymptotics in Magnetic Weyl Calculus},
	pdfkeywords={Magnetic field, quantization, pseudodifferential operator, Weyl calculus, Weyl product, asymptotic expansion, gauge invariance, small parameters, Dirac equation, 35S05, 47A60, 81Q05, 81Q10, 81Q15}
}
\maketitle

\begin{abstract}
	This paper is concerned with small parameter asymptotics of magnetic quantum systems. In addition to a semiclassical parameter $\eps$, the case of small coupling $\lambda$ to the magnetic vector potential naturally occurs in this context. Magnetic Weyl calculus is adapted to incorporate both parameters, at least one of which needs to be small. 
	Of particular interest is the expansion of the Weyl product which can be used to expand the product of operators in a small parameter, a technique which is prominent to obtain perturbation expansions. Three asymptotic expansions for the magnetic Weyl product of two Hörmander class symbols are proven: (i) $\eps \ll 1$ and $\lambda \ll 1$, (ii) $\eps \ll 1$ and $\lambda = 1$ as well as (iii) $\eps = 1$ and $\lambda \ll 1$. Expansions (i) and (iii) are impossible to obtain with ordinary Weyl calculus. 
	Furthermore, I relate results derived by ordinary Weyl calculus with those obtained with magnetic Weyl calculus by one- and two-parameter expansions. 
	To show the power and versatility of magnetic Weyl calculus, I derive the semirelativistic Pauli equation as a scaling limit from the Dirac equation up to errors of $4$th order in $\nicefrac{1}{c}$.  
\end{abstract}
\medskip

\smallskip


\tableofcontents

\section{Introduction} 
\label{intro}

Quantum mechanical systems often contain small parameters that allow us to order terms by magnitude and importance. One prominent example are adiabatic systems where the fast degrees adjust `instantaneously' to the configuration of the slow degrees of freedom. Here, the small parameter quantifies the separation of slow and fast scales. Under certain conditions, effective dynamics may be derived which contain corrections order-by-order in the small parameter and one can bound the error. This effective hamiltonian may be the starting point for a semiclassical analysis: Egorov-type theorems compare the quantization of suitable classically evolved observables with the corresponding time-evolved quantum observables. 

A quantization procedure is a systematic way to associate operators to functions on symplectic manifolds that has certain natural properties (\eg linearity and compatibility with the involution, see Chapter~5 \cite{Waldmann:deformationQuantization:2008} for an overview). Mathematically speaking, we are interested in a functional calculus for non-commuting observables called position $x$ and momentum $\xi$ on phase space $\PSpace \cong \R^d \times \R^d$ endowed with the \emph{magnetic} symplectic form. This is by no means the only interesting case, other examples are spin systems \cite{VarillyGraciaBondia:MoyalRepSpin:1989,VarillyGraciaBondiaSchempp:MoyalRepSpin:1990} or quantization procedures on generic Poisson manifolds  \cite{Kontsevich:defQ_Poisson_manifolds:1997,Waldmann:deformationQuantization:2008}. Before we explain magnetic quantization in detail, we will briefly recall the non-magnetic case. 

\medskip
\noindent
Usual Weyl quantization $\Ope : f \mapsto \Ope(f)$ maps suitable functions $f$ on phase space $\PSpace$ onto linear operators acting on (subspaces of) $L^2(\R^d)$ (see \cite{Hoermander:WeylCalculus:1979,Folland:harmonicAnalysisPhaseSpace:1989}, for example). The index $\eps$ indicates that the commutator of momentum $\Ope(\xi)$ and position $\Ope(x)$ is of order $\eps$, $i \bigl [ \Ope(\xi_l) , \Ope(x_j) \bigr ] = \eps \delta_{lj}$. With this quantization procedure in hand, it turns out we can \emph{define} a \emph{non-commutative} product~$\Weyle$ on phase space which emulates the operator product, 
\begin{align*}
	\Ope(f) \, \Ope(g) &= \Ope (f \Weyle g) \asymp \sum_{n = 0}^{\infty} \eps^n \, \Ope \bigl ( (f \Weyle g)_{(n)} \bigr ) + \Ope \bigl ( \ordere{\infty} \bigr ) 
	. 
\end{align*}
If $\eps \ll 1$, we can expand the Weyl product asymptotically in powers of $\eps$ up to an arbitrarily small error which allows us to rewrite the operator product as an asymptotic series in $\eps$ as well \cite{Folland:harmonicAnalysisPhaseSpace:1989}. This idea has been used to derive effective operators whose dynamics approximates the full perturbed dynamics, see, for instance, \cite{LittlejohnWeigert:diagonalizationMultiWave:1993,PST:sapt:2002,PST:effDynamics:2003,Teufel:adiabaticPerturbationTheory:2003}. Hence, from a computational point of view, an asymptotic expansion in a small parameter is a very desirable thing to have. 

\medskip
\noindent
Very often, the magnetic field $B$ \emph{is} the perturbation of the hamiltonian and usual (\ie non-magnetic) Weyl calculus is not well-adapted to this situation. A \emph{magnetic} Weyl calculus is needed and this paper is concerned with the derivation of three asymptotic expansions of the corresponding magnetic Weyl product, \eg with respect to a semiclassical parameter $\eps$, 
\begin{align*}
	f \Weyle^B g \asymp \sum_{n = 0}^{\infty} \eps^n \, (f \Weyle^B g)_{(n)} + \ordere{\infty}
	. 
\end{align*}
Let us introduce some notation first: assume we apply a magnetic field, thought of as a two-form $B = \bigl ( B_{lj} \bigr )_{1 \leq l,j \leq d}$, whose components are smooth, bounded and have bounded derivatives to all orders, \ie $B_{lj} \in \mathcal{BC}^{\infty}(\R^d)$, $1 \leq l,j \leq d$. Then we will consider the quantization which takes the position and momentum vectors $x$ and $\xi$ into the operators 
\begin{align}
	Q &:= \hat{x} \label{intro:observablesUsualScaling} \\
	P^A_{\eps,\lambda} &:= - i \eps \nabla_{x} - \lambda A(Q) \notag 
\end{align}
$A$ is a vector potential which \emph{represents} $B$, \ie $B_{lj} = \partial_{x_l} A_j - \partial_{x_j} A_l$, $1 \leq l,j \leq d$, whose components will always be chosen to be smooth and polynomially bounded, $A_l \in \mathcal{C}^{\infty}_{\mathrm{pol}}(\R^d)$, $1 \leq l \leq d$. The first parameter $\eps$ formally takes the role of $\hbar$ and quantifies the non-commutativity of position and momentum. There are many systems (\eg Born-Oppenheimer-type systems) where $\eps$ is, physically speaking, \emph{not} $\hbar$, but some other parameter that describes a separation of scales. The second parameter, $\lambda$, is often physically equal to $\nicefrac{e}{c}$ where $e$ is the charge quantum and $c$ the speed of light. If $\lambda$ is taken to be small, this could either mean that we are interested in some non-relativistic limit, $c \gg 1$, or in the limit of small charge $e \ll 1$. To be precise, small parameters must not have any units, so instead of $\lambda = \nicefrac{1}{c}$ we should really use $\lambda = \nicefrac{v_0}{c}$ where $v_0$ is some characteristic velocity. 

The commutators of our basic building blocks are given by 
\begin{align*}
	i [ Q_l , Q_j ] = 0 , && 
	i [ P^A_{\eps,\lambda \, j} , Q_l ] = \eps \delta_{lj} , && 
	i [ P^A_{\eps,\lambda \, l} , P^A_{\eps,\lambda \, j} ] = - \eps \lambda B_{lj}(Q) 
	. 
\end{align*}
We note that the commutator of the kinetic momentum operator depends on the magnetic field $B$ and \emph{not} on the specific choice of vector potential. If we translate these commutation relations to the classical framework where $x$ corresponds to $Q$ and $\xi$ to $P^A_{\eps,\lambda}$, we have to use the \emph{magnetic} symplectic form 
\begin{align*}
	\omega^{\lambda B} = \sum_{l = 1}^d \dd \xi_l \wedge \dd x_l - \tfrac{1}{2} \lambda \, \sum_{l,j = 1}^d B_{lj}(x) \, \dd x_l \wedge \dd x_j 
\end{align*}
on phase space $\PSpace$ which induces the \emph{magnetic} Poisson bracket, 
\begin{align*}
	\bigl \{ f , g \bigr \}_{\lambda B} = \sum_{l = 1}^d \bigl ( \partial_{\xi_l} f \, \partial_{x_l} g - \partial_{x_l} f \, \partial_{\xi_j} g \bigr ) - \lambda \sum_{l,j = 1}^d B_{lj}(x) \, \partial_{\xi_l} f \, \partial_{\xi_j} g 
	. 
\end{align*}
Classically, this agrees with the recipe of minimal substitution: if we replace $\xi$ with $\xi - \lambda A(x)$ and use the standard Poisson bracket $\{ \cdot , \cdot \}$, we recover the \emph{magnetic} Poisson bracket, 
\begin{align*}
	\bigl \{ f(x,\xi - \lambda A(x)) , g(x,\xi - \lambda A(x)) \bigr \} &= \bigl \{ f , g \bigr \}_{\lambda B} (x,\xi - \lambda A(x)) 
	. 
\end{align*}
Quantum mechanically, these two points of view are \emph{no longer equivalent}. 
Based on the \emph{magnetic} symplectic form, Müller was the first to define covariant magnetic Weyl calculus in a non-rigorous fashion \cite{Mueller:productRuleGaugeInvariantWeylSymbols:1999} (the author thanks R.~Littlejohn for this reference), although Luttinger has used it in prototypical form as early as 1951 \cite{Luttinger:magnetic_field_periodic_potential:1951}. The present paper relies upon earlier contributions which have put these ideas on a solid mathematical foundation \cite{IftimieMantiouPurice:magneticPseudodifferentialOperators:2005,MantoiuPurice:magneticWeylCalculus:2004,MantiouPurice:strictDeformationQuantizationMagneticField:2005,KarasevOsborn:symplecticArea:2001,KarasevOsborn:quantumMagneticAlgebra:2004,KarasevOsborn:cotangentBundleQ:2005}. Some notable results include a Caldéron-Vaillancourt-type theorem ($L^2$ continuity of $\Hoerrd{0}$ symbols), selfadjointness of elliptic symbols on magnetic Sobolev spaces \cite{IftimieMantiouPurice:magneticPseudodifferentialOperators:2005} and a Beals-type criterion \cite{MantoiuPurice:BealsCriterion:2008}. One missing ingredient is an asymptotic expansion of the magnetic product with respect to a small parameter. 

Before we continue, we would like to elaborate on possible choices of scalings. If we rescale space by $\eps$ via $({U_{\eps}}^{-1} \varphi)(x) := \eps^{\nicefrac{d}{2}} \, \varphi (\eps x)$, $\varphi \in L^2(\R^d)$, we can transform the observables~\eqref{intro:observablesUsualScaling} into 
\begin{align}
	\Qe &:= \eps \hat{x} \label{intro:observablesAdiabaticScaling} \\
	\PiAel &:= - i \nabla_x - \lambda A(\Qe) 
	\notag 
	. 
\end{align}
Mathematically, both scales are unitarily equivalent (see Appendix~\ref{appendix:equivalenceWeylSystems}). The decision which scale is deemed preferable is based on the physics of the problem. If we would like to emphasize the slow variation of the magnetic field (compared to other potentials), then the second choice is more natural. The single-particle Schrödinger equation with periodic potential $V_{\Gamma}$ subjected to a slowly-varying electromagnetic field, a system which is described by the hamiltonian 
\begin{align*}
	H = \tfrac{1}{2} \bigl ( - i \nabla_x - A(\eps \hat{x}) \bigr )^2 + V_{\Gamma}(\hat{x}) + \Phi(\eps \hat{x}) , 
\end{align*}
falls in this category. We emphasize that all of our results hold in \emph{either scaling}. In particular, the asymptotic expansion of the product is the \emph{same}, independent of the choice of scaling (see Appendix \ref{appendix:equivalenceWeylSystems} for details). As this paper was initially motivated by the problem above, we will use the \emph{adiabatic} scaling given by equation~\eqref{intro:observablesAdiabaticScaling}. 

The fundamental building block of magnetic pseudodifferential calculus is the \emph{magnetic Weyl system} that implements the commutation relations, 
\begin{align*}
	\WeylSysAel(X) := e^{- i \sigma(X,(\Qe,\PiAel))} 
	. 
\end{align*}
Its gauge-covariance leads to the gauge-covariance of magnetic pseudodifferential operators. Here $(x,\xi) = X \in \PSpace$ is a point in phase space and $\sigma(X,Y) := \xi \cdot y - x \cdot \eta$, $Y = (y,\eta) \in \PSpace$, is the (non-magnetic) symplectic form. In \cite{MantoiuPurice:magneticWeylCalculus:2004} it has been shown that this is a well-defined operator which acts on any $\varphi \in L^2(\R^d)$ by 
\begin{align*}
	(\WeylSysAel(Y) \varphi)(x) &= e^{- i \eps (x + \nicefrac{y}{2}) \cdot \eta} e^{-i \lambda \GAe([x,x+y])} \, \varphi(x+y) 
	. 
\end{align*}
The appearance of the magnetic circulation $\GAe$ along the line segment which connects $x$ and $y$ 
\begin{align}
	\lambda \GAe([x,y]) := 
	\frac{\lambda}{\eps} \int_{[\eps x , \eps y]} A 
	\label{intro:scaledCirculation}
\end{align}
stems from the use of magnetic translations. All proofs in \cite{IftimieMantiouPurice:magneticPseudodifferentialOperators:2005,MantoiuPurice:magneticWeylCalculus:2004,MantiouPurice:strictDeformationQuantizationMagneticField:2005} carry over to the present case via a simple scaling argument. The pseudodifferential operator associated to a Schwartz function $f \in \Schwartz(\PSpace)$ is defined in terms of the symplectic Fourier transform $\Fs f = \Fs^{-1} f$ and the Weyl system: 
\begin{align}
	\OpAel (f) &:= 
	\frac{1}{(2\pi)^{d}} \int \dd X \, (\Fs^{-1} f)(X) \, \WeylSysAel(X) 
	:= \frac{1}{(2\pi)^{2d}} \int \dd X \, \int \dd \tilde{X} \, e^{i \sigma(X,\tilde{X})} f(\tilde{X}) \, \WeylSysAel(X) 
\end{align}
All parameters are contained in the Weyl system $\WeylSysAel(X) = e^{- i \sigma(X,(\Qe,\PiAel))}$; if we had chosen the usual scaling, the formula would be the same, but $\Qe$ and $\PiAel$ would have to be replaced by $Q$ and $P^A_{\eps,\lambda}$ from equation~\eqref{intro:observablesUsualScaling}. This definition can be extended to observables of Hörmander symbol class $m$ with weight $\rho \in [0,1]$ \cite{MantoiuPurice:magneticWeylCalculus:2004,IftimieMantiouPurice:magneticPseudodifferentialOperators:2005} among others: 
\begin{align*}
	\Hoerrd{m} := \Bigl \{ f \in \mathcal{C}^{\infty}(\PSpace) \; \big \vert \; \forall a , \alpha \in \N_0^d \, \exists C_{a \alpha} > 0 : \; \babs{\partial_x^a \partial_{\xi}^{\alpha} f(x,\xi)} \leq C_{a \alpha} \expval{\xi}^{m - \sabs{\alpha} \rho} \Bigr \}
\end{align*}
The quantization of Hörmander-class symbols $f$ act on Schwartz functions $\varphi \in \Schwartz(\R^d)$ as 
\begin{align*}
	(\OpAel (f) \varphi)(x) &= \frac{1}{(2 \pi)^d} \int \dd y \int \dd \eta \, e^{- i (y - x) \cdot \eta} e^{- i \lambda \GAe([x,y])} \, f \bigl ( \tfrac{\eps}{2} (x+y) , \eta \bigr ) \, \varphi(y) 
\end{align*}
where the integral is interpreted as an oscillatory integral. If in addition the symbol $f$ is elliptic of order $m$, for instance, this definition extends from $\varphi \in \Schwartz(\R^d)$ to functions in the (magnetic) Sobolev space $H^m_A(\R^d)$. These results are well-known for standard Weyl calculus, see \eg \cite{Hoermander:fourierIntegralOperators1:1972,DHoermander:fourierIntOp2:1973,Hoermander:WeylCalculus:1979,Hoermander:analysisPDO3:1985,Stein:harmonicAnalysis:1993}. If we choose an equivalent gauge, \ie $A'(x) = A(x) + \eps \nabla_x \chi(x)$ for some $\chi \in \mathcal{C}^{\infty}_{\mathrm{pol}}(\R^d)$, then the magnetic Weyl quantization of $f$ with respect to $A'$ is related to that with respect to $A$ by conjugating with $e^{+i \lambda \chi(\Qe)}$, that is, magnetic quantization is \emph{covariant}: 
\begin{align*}
	\Op^{A + \eps \nabla_x \chi}_{\eps,\lambda}(f) = e^{+i \lambda \chi(\Qe)} \, \OpAel(f) \, e^{-i \lambda \chi(\Qe)} 
\end{align*}
\emph{Unless} $f$ is a polynomial of degree $\leq 2$ in momentum, regular Weyl quantization of minimally substituted symbols is \emph{not gauge covariant} and does not coincide with magnetic quantization. This includes physically relevant examples such as $\sqrt{m^2 + \xi^2}$ or band energy functions in solid state physics. 

\medskip
\noindent
The second major component in magnetic Weyl calculus is a product $\magWel$. Its form is shaped by the commutation relations of the fundamental observables as expressed by the composition law of the Weyl system, 
\begin{align}
	\WeylSysAel(Y) \, \WeylSysAel(Z) &= e^{i \frac{\eps}{2} \sigma(Y,Z)} \, e^{- i \frac{\lambda}{\eps} \Gamma^{B}(\sexpval{\Qe, \Qe + \eps y , \Qe + \eps y + \eps z})} \, \WeylSysAel(Y+Z) 
	. 
	\label{intro:compWeylSystem}
\end{align}
The magnetic contribution $e^{- i \frac{\lambda}{\eps} \Gamma^{B}(\sexpval{\Qe, \Qe + \eps y , \Qe + \eps y + \eps z})}$ is the exponential of the magnetic flux through the triangle with corners $\Qe$, $\Qe + \eps y$ and $\Qe + \eps y + \eps z$ (see equation \eqref{asympExp:magneticFlux}) and as such depends only on the \emph{magnetic field} and not on the choice of gauge. A scaled magnetic flux $\gBe$ through triangles with different corners (equation~\eqref{asympExp:eqn:mag_flux_product_corners}) enters the integral formula for the magnetic Weyl product, 
\begin{align}
	(f \magWel g)(X) &= \frac{1}{(2 \pi)^{2d}} \int \dd Y \, \int \dd Z \, e^{+i \sigma(X,Y+Z)} \, e^{i \frac{\eps}{2} \, \sigma(Y,Z) - i \lambda \gBe(x,y,z)} \, \bigl ( \Fs^{-1} f \bigr )(Y) \, \bigl ( \Fs^{-1} g \bigr )(Z) 
	. 
	\label{intro:eqn:mag_Weyl_product}
\end{align}
This expression can be derived from the composition law of the Weyl system and the magnetic Wigner transform (see~Theorem~\ref{asympExp:thm:equivalenceProduct}). If $\eps \ll 1$ and $y$ and $z$ have lengths of $\orderone$, then $\gBe(x,y,z)$ is of order $\ordere{}$ and can be expanded in powers of $\eps$. The different asymptotic expansions with respect to $\eps$ and $\lambda$ simultaneously, $\eps$ or $\lambda$ are obtained by expanding exponential of the `twister' 
\begin{align*}
	e^{i T_{\eps,\lambda}(x,Y,Z)} := e^{i \frac{\eps}{2} \sigma(Y,Z) - i \lambda \gBe(x,y,z))} 
\end{align*}
in the appropriate way. One can easily distinguish purely magnetic and purely non-magnetic contributions in the product. 

In the `usual scaling' (equation \eqref{intro:observablesUsualScaling}), the Weyl system would essentially obey the same composition law and lead to the \emph{exact same} expansion of the magnetic product (see Theorem~\ref{appendix:equivalence_expansion} for details).

\subsection{Comparison with usual Weyl calculus} 
\label{intro:main_result:comparisonNonMagWQ}

Regular Weyl quantization has seen many applications to magnetic systems over the years, so one obvious question is how results would differ if magnetic Weyl calculus had been used instead. In ordinary Weyl calculus, a magnetic field is included by quantizing the minimally substituted symbol $h \circ \vartheta^{\lambda A}$ instead of $h$ where $\vartheta^{\lambda A}(x,\xi) := \bigl ( x , \xi - \lambda A(x) \bigr )$. If $h \in \Hoerrd{m}$ is a Hörmander symbol, then one typically assumes that the magnetic field $B$ is such that there exists a vector potential whose components of $A$ are $\mathcal{BC}^{\infty}(\R^d)$ functions or that all derivatives are bounded, $\sabs{\partial_x^a A(x)} \leq C_a$ for some $C_a > 0$ and all $a \in \N_0^d$, $\abs{a} \geq 1$. The latter is typically used to cover the case of constant magnetic field. Under these assumptions, Theorem~\ref{asympExp:thm:usualWQmagWQ} ensures the difference 
\begin{align}
	\Ope(h \circ \vartheta^{\lambda A}) -  \Op^A_{\eps,\lambda}(h) = \ordere{2} 
	\label{intro:main_result:comparisonNonMagWQ:eqn:comparison_quantizations}
\end{align}
is a pseudodifferntial operator which is the quantization of some $g \in \Hoerrd{m-3}$. The difference can be expressed as a power expansion in $\eps$. 

The second important ingredient is a comparison of the two products: standard Weyl calculus suggests to multiply $f \circ \vartheta^{\lambda A}$ and $g \circ \vartheta^{\lambda A}$ where $f \in \Hoerrd{m_1}$ and $g \in \Hoerrd{m_2}$. Combining 
\begin{align*}
	\bigl \{ f \circ \vartheta^{\lambda A} , g \circ \vartheta^{\lambda A} \bigr \} = \bigl \{ f , g \bigr \}_{\lambda B} \circ \vartheta^{\lambda A}
\end{align*}
with equation~\eqref{intro:main_result:comparisonNonMagWQ:eqn:comparison_quantizations}, we also get that the two products differ by $\ordere{2}$, 
\begin{align*}
	\Ope \bigl ( f \circ \vartheta^{\lambda A} \Weyle g \circ \vartheta^{\lambda A} \bigr ) = \OpAel \bigl ( f \magWel g \bigr ) + \ordere{2} 
	. 
\end{align*}
Hence, first-order corrections (\eg the subprincipal symbol of an effective hamiltonian) derived with usual Weyl calculus coincide with those obtained by magnetic Weyl calculus. However, the advantages of magnetic Weyl calculus are significant, most importantly do results extend to much more general magnetic fields, one only needs to assume that the components of $B$ are in $\mathcal{BC}^{\infty}(\R^d)$. The gauge-covariance of magnetic Weyl calculus implies that results depend on the properties of the magnetic field rather than those of the vector potential. 


\subsection{Main results} 
\label{intro:main_result}

The main result of this work is Theorem \ref{intro:thm:main_result:asymptoticExpansion} which gives an asymptotic \emph{two}-parameter expansion of the product of two Hörmander class symbols $f \in \Hoerrd{m_1}$ and $g \in \Hoerrd{m_2}$ where each term involves the components of the magnetic field $B$, $f$ and $g$ and all its derivatives evaluated at the same point $(x,\xi)$. Furthermore, two one-parameter expansions have been derived: for $\eps \ll 1$, the expansion still has the same structure as the two-parameter expansion. In case $\eps$ is not necessarily small, one can expand the product with respect to $\lambda \ll 1$ and the terms are given by oscillatory integrals. 
\begin{thm}[Asymptotic expansion of the magnetic Moyal product]\label{intro:thm:main_result:asymptoticExpansion}
	Assume $B$ is a magnetic field whose components are $\mathcal{BC}^{\infty}$ functions and $f \in \Hoerrd{m_1}$ as well as $g \in \Hoerrd{m_2}$. Then the magnetic Moyal product can be expanded asymptotically in $\eps \ll 1$ and $\lambda \ll 1$: for every $\eprec \ll 1$ we can choose $N \equiv N(\eprec,\eps,\lambda) \in \N_0$ such that 
	\begin{align}
		f \magWel g &= \sum_{n = 0}^{N} \sum_{k = 0}^n \eps^n \lambda^k \, (f \magWel g)_{(n,k)} + \tilde{R}_N 
		\label{intro:epsLambdaExpansion}
	\end{align}
	where the $(n,k)$ term $(f \magWel g)_{(n,k)}$ is in symbol class $\Hoerrd{m_1 + m_2 - (n+k) \rho}$ and we have explicit control over the remainder: $\tilde{R}_N$ as given by equation \eqref{asympExp:eqn:total_remainder} is numerically small and in the correct symbol class, $\Hoerrd{m_1 + m_2 - (N+1)  \rho}$, \ie it is of order $\ordereprec$ in the sense of Definition \ref{asympExp:defn:precision}. The $(n,k)$ term of the expansion, 
	\begin{align}
		(f \magWel g)_{(n,k)}(X) &= \sum_{\substack{k_0 + \sum_{j = 1}^{n} j k_j = n \notag \\
		\sum_{j=1}^{n} k_j = k}} \frac{i^{k + k_0}}{k_0! \, k_1 ! \cdots k_{n}!} 
		\cdot \\
		&\qquad \cdot 
		{\mathcal{L}_0}^{k_0} \bigl ( (\partial_{\eta},\partial_y) , (\partial_{\zeta},\partial_z) \bigr ) \, \prod_{j = 1}^{n} {\mathcal{L}_j}^{k_j}(x,-i \partial_{\eta},-i \partial_{\zeta}) \bigr ) f(Y) g(Z) \Bigl . \Bigr \rvert_{Y = X = Z} \notag 
		, 
	\end{align}
	is defined in terms of a family of differential operators $\mathcal{L}_j$, $j \in \N_0$, 
	\begin{align}
		\mathcal{L}_0(Y,Z) &:= \tfrac{1}{2} \sigma(Y,Z) = \tfrac{1}{2} \bigl ( \eta \cdot z - y \cdot \zeta \bigr ) \\
		\mathcal{L}_j(x,y,z) &:= - \frac{1}{j!} \, \sum_{m_1 , \ldots , m_{j-1} = 1}^d \partial_{x_{m_1}} \cdots \partial_{x_{m_{j-1}}} B_{kl}(x) \, y_k \, z_l \, \left ( - \frac{1}{2} \right )^{j+1} \frac{1}{(j+1)^2} \sum_{c = 1}^j \noverk{j+1}{c} \, \cdot \notag \\
		&\qquad \qquad \qquad \qquad \cdot \bigl ( (1 - (-1)^{j+1}) c - (1 - (-1)^{c}) (j+1) \bigr ) \, y_{m_1} \cdots y_{m_{c-1}} z_{m_{c}} \cdots z_{m_{j-1}} \notag \\
		&=: - \sum_{\abs{\alpha} + \abs{\beta} = j-1} C_{j,\alpha,\beta} \, \partial_x^{\alpha} \partial_x^{\beta} B_{kl}(x) \, y_k z_l \, y^{\alpha} \, z^{\beta} 
		. 
	\end{align}
\end{thm}
Here, we have glossed over the difficulty of agreeing up to which order we have to expand the product ($\eps$ and $\lambda$ are \emph{independent}), as we can no longer use well-known notation such as $\ordere{n}$ or $\orderl{k}$ without modification. We refer to Section \ref{asympExp:semiclSymbolsPrecision} for details. 
\begin{remark}
	With that in mind, we can give the first terms of the expansion concisely as 
	\begin{align*}
		(f \magWel g)_{(0,0)} &= f \, g 
		, 
		\\
		(f \magWel g)_{(1,0)} &= - \tfrac{i}{2} \sum_{l = 1}^d \bigl ( \partial_{\xi_l} f \, \partial_{x_l} g - \partial_{x_l} f \, \partial_{\xi_l} g \bigr ) 
		, 
		\\
		(f \magWel g)_{(1,1)} &= + \tfrac{i}{2} \sum_{l,j = 1}^d B_{lj} \, \partial_{\xi_l} f \,  \partial_{\xi_j} g 
		. 
	\end{align*}
	The second-order corrections contain at least two derivatives with respect to momentum; if we group by powers of $\eps$, then the decay properties are determined by $(f \magWel g)_{(2,0)}$, 
	\begin{align*}
		(f \magWel g)_{(2,0)} &= - \tfrac{1}{8} \sum_{l,j = 1}^d \bigl ( 
		\partial_{\xi_l} \partial_{\xi_j} f \,  \partial_{x_l} \partial_{x_j} g + 
		\partial_{x_l} \partial_{x_j} f \,  \partial_{\xi_l} \partial_{\xi_j} g 
			+ \bigr . \\
			\bigl . &\qquad \qquad \qquad \qquad - 
		\partial_{\xi_l} \partial_{x_j} f \,  \partial_{x_l} \partial_{\xi_j} g - 
		\partial_{x_l} \partial_{\xi_j} f \,  \partial_{\xi_l} \partial_{x_j} g \bigr ) 
		, 
		\\
		(f \magWel g)_{(2,1)} &= + \tfrac{i}{4} \sum_{l,j,k = 1}^d \Bigl ( \tfrac{1}{6} \partial_{x_j} B_{lk} \, \bigl ( 
		\partial_{\xi_l} \partial_{\xi_j} f \, \partial_{\xi_k} g -
		\partial_{\xi_l} f \,                  \partial_{\xi_j} \partial_{\xi_k} g \bigr ) + \Bigr . \\
		\Bigl . &\qquad \qquad \qquad \qquad - 
		B_{lk} \, \bigl ( 
		\partial_{\xi_l} \partial_{\xi_j} f \, \partial_{\xi_k} \partial_{x_j} g - 
		\partial_{\xi_l} \partial_{x_j} f \, \partial_{\xi_k} \partial_{\xi_j} g  \bigr ) 
		\Bigr ) 
		, 
		\\
		(f \magWel g)_{(2,2)} &= - \tfrac{1}{8} \sum_{l_1,l_2,j_1,j_2 = 1}^d B_{l_1 j_1} \, B_{l_2 j_2} \, \partial_{\xi_{l_1}} \partial_{\xi_{l_2}} f \,  \partial_{\xi_{j_1}} \partial_{\xi_{j_2}} g 
		. 
	\end{align*}
	If the magnetic field is constant, all terms containing derivatives of $B$ vanish. 
\end{remark}
For each order in $\eps$, the sum in $\lambda$ is \emph{finite} and we immediately obtain the 
\begin{cor}[Expansion of $\magWel$ in $\eps$]
	If the assumptions of Theorem~\ref{intro:thm:main_result:asymptoticExpansion} are modified by taking $\lambda = 1$, then the $\eps$ expansion of the product $f \magWel g$ of two Hörmander symbols can be recovered from the two-parameter expansion: the $n$th order term in $\eps$ then reads 
	\begin{align*}
		(f \magWel g)_{(n)} &= \sum_{k = 0}^n \, (f \magWel g)_{(n,k)} \in \Hoerrd{m_1 + m_2 - n \rho} 
	\end{align*}
	where the $(f \magWel g)_{(n,k)}$ are taken from Theorem~\ref{intro:thm:main_result:asymptoticExpansion}. 
\end{cor}
If $\eps$ is \emph{not} small, we can no longer expand the magnetic flux integral as a Taylor series in $\eps$, but have to accept it as-is. Only part of the twister can be expanded and the terms in the $\lambda$ expansion are obtained from replacing $e^{i T_{\eps,\lambda}(x,Y,Z)}$ with 
\begin{align*}
	e^{i T_{\eps,\lambda}(x,Y,Z)} = e^{i \frac{\eps}{2} \sigma(Y,Z)} \, e^{-i \lambda \gBe (x,y,z)} \asymp e^{i \frac{\eps}{2} \sigma(Y,Z)} \, \bigl ( 1 - i \lambda \gBe (x,y,z) + \orderl{2} \bigr ) 
\end{align*}
in equation~\eqref{intro:eqn:mag_Weyl_product}. The integral formulas cannot be simplified any further unless the symbols which are multiplied with each other have a special structure (\eg when they are polynomials in $x$ or $\xi$). 
\begin{thm}[Expansion of $\magWel$ in $\lambda$]\label{intro:thm:lambdaExpansion}
	Assume the magnetic field $B$ has components of class $\mathcal{BC}^{\infty}$. Then for $\lambda \ll 1$ and $\eps \leq 1$, we can expand the $\lambda$ Weyl product of $f \in \Hoerrd{m_1}$ and $g \in \Hoerrd{m_2}$ asymptotically in $\lambda$ such that 
	\begin{align}
		f \magWel g - \sum_{k = 0}^N \lambda^k (f \magWel g)_{(k)} \in \Hoerrd{m_1 + m_2 - 2(N+1) \rho} 
		, 
		&& 
		(f \magWel g)_{(k)} \in \Hoerrd{m_1 + m_2 - 2k \rho} 
		. 
	\end{align}
	In particular, the zeroth-order term reduces to the \emph{non-magnetic} Weyl product, $(f \magWel g)_{(0)} = f \Weyle g$. We have explicit control over the remainder (equation~\eqref{aysmpExp:eqn:lambda_expansion:remainder}): if we expand the product up to $N$th order in $\lambda$, the remainder is of order $\orderl{N+1}$ and in symbol class $\Hoerrd{m_1 + m_2 - 2(N+1)\rho}$. 
\end{thm}
The equivalence of the $\eps \rightarrow \lambda$ expansion to the $\lambda \rightarrow \eps$ expansion is obtained through explicit computation in Section \ref{asympExp:equivalenceProduct}. Agreeing on a remainder is somewhat tricky and necessitated the introduction of the concept of \emph{precision} (Defintion \ref{asympExp:defn:precision}), because the numerical values of $\eps$ and $\lambda$ vary independently. 
\begin{thm}\label{intro:thm:equivalenceExpansions}
	Under the assumptions of Theorem~\ref{intro:thm:main_result:asymptoticExpansion}, the magnetic Weyl product $\magWel$ of two symbols $f \in \Hoerrd{m_1}$, $g \in \Hoerrd{m_2}$ can be \emph{simultaneously} expanded in $\eps$ and $\lambda$, \ie the expansion is the same, regardless of whether we expand with respect to $\eps$ first and then $\lambda$ or the other way around. 
\end{thm}
%


\subsection{Structure} 
\label{intro:structure}

The derivation of our main results are found in Section~\ref{asympExp}: before we derive the main result, we need some prerequisites. First, the notions of two-parameter symbol classes and precision are introduced (Section~\ref{asympExp:semiclSymbolsPrecision}). The properly adapted Wigner transform (Section~\ref{asympExp:magWTransform}) is necessary to show the equivalence of two versions of the product formula found in the literature (Section~\ref{asympExp:equivalenceProduct}). The one that is more amenable to an asymptotic expansion is used to derive the main result in Section~\ref{asympExp:asymptoticExp}. Lastly, we relate magnetic and non-magnetic quantization in Section~\ref{asympExp:magWQminSub} to be able to connect results derived via regular Weyl calculus to those where magnetic Weyl calculus has been used. 

As a simple, but non-trivial application, the semirelativistic Pauli equation is derived from the Dirac equation (Section~\ref{Dirac}). It illustrates the versatility of the two-parameter expansion and gives insight into the origin of the corrections. We place emphasis on the mechanics of the computation. For the sake of brevity, the example is not presented in a mathematically rigorous manner, this is postponed to a future publication \cite{FuerstLein:nonRelLimitDirac:2008}. 

In an attempt to clean up the presentation, we have moved some auxiliary technical lemmas and details of various straightforward, but tedious calculations to an Appendix. 


\subsection{Acknowledgements} 
\label{intro:acknowledgements}

This work was supported by a DAAD scholarship. I thank R.~Littlejohn and R.~Purice for their kind hospitality. I am very grateful for useful discussions, insights and references from M.~M\u antoiu, T.~Miyao, G.~Panati and H.~Spohn. 


\section{Asymptotic expansion in $\lambda$ and $\eps$} 
\label{asympExp}

This section will contain the proofs to my main results, namely the two-parameter expansion and some theorems which connect magnetic and non-magnetic Weyl calculus. Before we can attend to the asymptotic expansion, we need some preliminaries: apart from assumptions on the magnetic field and some comments on the notation, we need to introduce the concept of precision as well as adapt the definition of the Wigner-Weyl transform. 
\medskip

\noindent
For simplicity, we will use Einstein's summation convention throughout this paper, \ie \textbf{repeated indices in a product are always summed over from $1$ to $d$}. We will always assume that the magnetic field satisfies the following assumptions unless explicitly stated otherwise. 
\begin{assumption}\label{asympExp:assumption:usualConditionsBA}
	We assume that the components of the magnetic field $B = \dd A$ and associated vector potentials $A$ satisfy $B_{kl} \in \mathcal{BC}^{\infty}(\R^d,\R)$ and $A_l \in \mathcal{C}^{\infty}_{\mathrm{pol}}(\R^d,\R)$, respectively, for all $1 \leq k,l \leq d$. 
\end{assumption}
\begin{remark}
	If a magnetic field $B$ satisfies the above assumption, it is always possible to choose a polynomially bounded vector potential, \eg we may use the transversal gauge (equation~\eqref{appendix:eqn:transversal_gauge}). It is also clear that if $B$ and $A$ satisfy this assumption, then so do the scaled field $B^{\eps,\lambda}(x) := \dd \Ael(x) = \eps \lambda B(\eps x)$ and scaled potential $\Ael(x) := \lambda A(\eps x)$. 
\end{remark}
In magnetic Weyl quantization, magnetic circulations and flux integrals play a very prominent role. We define the circulation of the one-form $A$ along the line that connects $x$ and $y$ as 
\begin{align}
	\Gamma^A([x,y]) := \int_{[x,y]} A = (y - x) \cdot \int_0^1 \dd s \, A \bigl (x + s(y - x) \bigr ) 
	\label{asympExp:magneticCirculation}
	. 
\end{align}
The magnetic flux through the triangle with corners $x$, $y$ and $z$ (which we denote by $\sexpval{x,y,z}$) is the (gauge-invariant) integral of the magnetic two-form, 
\begin{align}
	\Gamma^B(x,y,z) := \int_{\sexpval{x,y,z}} B 
	\label{asympExp:magneticFlux}
	. 
\end{align}
Either we parametrize the triangle as in \cite{IftimieMantiouPurice:magneticPseudodifferentialOperators:2005} or we can choose a vector potential for $B = \dd A$ and use the Stoke's Theorem to write $\Gamma^B(x,y,z) = \Gamma^A([x,y]) + \Gamma^A([y,z]) + \Gamma^A([z,x])$. We will use the latter to derive the asymptotic expansion of the scaled flux integral 
\begin{align}
	\gBe(x,y,z) := \Gamma^B \Bigl ( \expval{x-\tfrac{\eps}{2}(y+z),x + \tfrac{\eps}{2}(y-z) , x+\tfrac{\eps}{2}(y+z)} \Bigr )
	\label{asympExp:eqn:mag_flux_product_corners}
\end{align}
in powers of $\eps$.

\subsection{Semiclassical symbols and precision} 
\label{asympExp:semiclSymbolsPrecision}

The Hörmander symbol classes $\Hoerrd{m}$ are Fréchet spaces whose topology can be defined by the usual family of seminorms 
\begin{align*}
	\norm{f}_{m a \alpha} := \sup_{(x,\xi) \in \PSpace} \expval{\xi}^{-m + \sabs{\alpha} \rho} \babs{\partial_x^a \partial_{\xi}^{\alpha} f(x,\xi)}  
	, 
	&& 
	a , \alpha \in \N_0^d 
	. 
\end{align*}
One important notion is that of a \emph{semiclassical symbol} \cite{PST:sapt:2002}, \ie it is a symbol which admits an expansion in $\eps$ and $\lambda$ which is in some sense uniform. 
\begin{defn}[Semiclassical two-parameter symbol]\label{asympExp:defn:2ParameterSemiSymbol}
	A map $f : [0,\eps_0) \times [0,\lambda_0) \longrightarrow \Hoerrd{m}$, $(\eps,\lambda) \mapsto f^{\eps,\lambda}$ is called \emph{semiclassical two-parameter symbol} of order $m$ with weight $\rho \in [0,1]$, if there exists a sequence $\{ f_{n,k} \}_{n,k \in \N_0}$, $f_{n,k} \in \Hoerrd{m - (n+k)\rho}$ for all $n,k \in \N_0$, such that 
	\begin{align*}
		f^{\eps,\lambda} - \sum_{l = 0}^{N} \sum_{n + k = l} \eps^n \lambda^k f_{n,k} \in \Hoerrd{m - (N + 1)\rho} && \forall N \in \N_0 
	\end{align*}
	uniformly in the following sense: for each $j \in \N_0$ there exists a constant $C_{N,m,j} > 0$ (independent of $\eps$ and $\lambda$) such that 
	\begin{align*}
		\Bigl \lVert f^{\eps,\lambda} - \sum_{l = 0}^{N} \sum_{n + k = l} \eps^n \lambda^k f_{n,k} \Bigr \rVert_{m,j} < C_{N,m,j} \, \max \{ \eps , \lambda \}^{N+1} . 
	\end{align*}
	holds for all $\eps \in [0,\eps_0)$ and $\lambda \in [0,\lambda_0)$. 
\end{defn}
Since $\eps$ and $\lambda$ vary independently, we also have to introduce a more sophisticated concept of precision. This is a technicality, but a definition is necessary to prove that expanding $f \magWel g$ first with respect to $\eps$ and then $\lambda$ yields the same asymptotics as when the product is expanded with respect to $\lambda$ and then with respect to $\eps$ (Theorem~\ref{asympExp:thm:lambdaEpsExpansion}). If there were only one small parameter, say $\eps$, then $f - g = \mathcal{O}(\eps^n)$ for symbols $f,g \in \Hoerrd{m}$ implies two things: (i) the difference between $f$ and $g$ is `numerically small' and (ii) we have associated a symbol class $\Hoerrd{m - n \rho}$ to the `number' $\eps^n$. In case of two independent parameters, such a simple concept will not do and we have to introduce an association between a \emph{third} number $\eprec \ll 1$ and a certain symbol class. Although it seems artificial at first to introduce yet another small parameter, in physical applications, this is quite natural: say, we are interested in the dynamics generated by a two-parameter symbol $H^{\eps,\lambda}$ on times of order $\mathcal{O}(\nicefrac{1}{\eprec})$, \ie $e^{- i \frac{t}{\eprec} \, H^{\eps,\lambda}}$. Then we need to include all terms in our expansion for which $\eps^n \lambda^k \leq \eprec$. Even if we choose $\eprec = \eps$, for instance, we still cannot avoid this abstract definition as $\lambda$ is independent of $\eps$. 
\begin{defn}[Precision $\ordereprec$]\label{asympExp:defn:precision}
	Let $\eps \ll 1$, $\lambda \ll 1$. For $\eprec \ll 1$, we define $n_c,k_c,N \in \N_0$ such that 
	\begin{align*}
		\eps^{n_c + 1} < \eprec \leq \eps^{n_c}
		, 
		&& 
		\lambda^{k_c + 1} < \eprec \leq \lambda^{k_c} 
	\end{align*}
	and $N \equiv N(\eps,\lambda,\eprec) := \max \{ n_c , k_c \}$. 
	We say that a finite resummation $\sum_{n = 0}^{N_{\eps}} \sum_{k = 0}^{N_{\lambda}} \eps^n \lambda^k f_{n,k}$ of a semiclassical symbol $f^{\eps,\lambda} \in \SemiHoerrd{m}$ is $\ordereprec$-close, 
	\begin{align*}
		f^{\eps,\lambda} - \sum_{n = 0}^{N_{\eps}} \sum_{k = 0}^{N_{\lambda}} \eps^n \lambda^k f_{n,k} = \ordereprec , 
	\end{align*}
	iff $f^{\eps,\lambda} - \sum_{n = 0}^{N_{\eps}} \sum_{k = 0}^{N_{\lambda}} \eps^n \lambda^k f_{n,k} \in \Hoerrd{m-(N+1)\rho}$ and $N_{\eps} , N_{\lambda} \geq N$. 
\end{defn}
%


\subsection{Magnetic Wigner transform} 
\label{asympExp:magWTransform}

The Wigner transform plays a central role because it can be used to relate states (density operators) to pseudo-probability measures on phase space. We will need it to show the equivalence of two integral formulas for the magnetic Weyl product $\magWel$. 
\begin{defn}[Magnetic Wigner transform]
	Let $\varphi,\psi \in \Schwartz(\R^d)$. The magnetic Wigner $\mathcal{W}^A(\varphi,\psi)$ is defined as 
	\begin{align*}
		\WignerTrafoel(\varphi,\psi)(X) := \eps^d \, \bigl ( \Fs \bscpro{\varphi}{\WeylSysAel(\cdot) \psi} \bigr )(-X) . 
	\end{align*}
\end{defn}
\begin{lem}\label{magBel:WignerTransformExpectation}
	The Wigner transform $\WignerTrafoel(\varphi,\psi)$ with respect to $\varphi,\psi \in \Schwartz(\R^d)$ is given by 
	\begin{align*}
		\WignerTrafoel(\varphi,\psi)(X) &= \int \dd y \, e^{- i y \cdot \xi} e^{-i \lambda \GAe([\nicefrac{x}{\eps} - \nicefrac{y}{2} ,\nicefrac{x}{\eps} + \nicefrac{y}{2} ])} \, \varphi^{\ast} \bigl ( \tfrac{x}{\eps} - \tfrac{y}{2} \bigr ) \, \psi \bigl ( \tfrac{x}{\eps} + \tfrac{y}{2} \bigr ) 
	\end{align*}
	and maps $\Schwartz(\R^d) \times \Schwartz(\R^d)$ unitarily onto $\Schwartz(\R^{2d})$. 
\end{lem}
\begin{proof}
	Formally, the result follows from direct calculation. 
	The second claim, $\WignerTrafoel(\varphi,\psi) \in \Schwartz(\R^{2d})$ follows from $e^{- i \lambda \GAe([\nicefrac{x}{\eps} - \nicefrac{y}{2} , \nicefrac{x}{\eps} + \nicefrac{y}{2}])} \,  \, \varphi^{\ast} \bigl ( \tfrac{x}{\eps} - \tfrac{y}{2} \bigr ) \, \psi \bigl ( \tfrac{x}{\eps} + \tfrac{y}{2} \bigr ) \in \Schwartz(\R^d \times \R^d)$ and the fact that the partial Fourier transformation is a unitary on $\Schwartz$. 
\end{proof}
\begin{remark}
	The Wigner transform can be easily extended to a map from $L^2(\R^d \times \R^d)$ into $L^2(\R^{2d}) \cap \mathcal{C}_{\infty}(\R^{2d})$ where $\mathcal{C}_{\infty}(\R^{2d})$ is the space of continuous functions which decay at $\infty$. For more details, see \cite[Proposition~1.92]{Folland:harmonicAnalysisPhaseSpace:1989}, for example. 
\end{remark}
Let $\mathcal{C}^{\infty}_{\mathrm{pol} \, u}(\R^{2d})$ be the space of smooth functions with uniform polynomial growth at infinity, \ie for each $f \in \mathcal{C}^{\infty}_{\mathrm{pol} \, u}(\R^{2d})$ we can find $m \in \R$, $m \geq 0$, such that for all multiindices $a , \alpha \in \N_0^d$ there is a $C_{a \alpha} > 0$ with 
\begin{align*}
	\babs{\partial_x^{a} \partial_{\xi}^{\alpha} f(x,\xi)} < C_{a \alpha} \expval{\xi}^m 
	, 
	&& 
	\forall (x,\xi) \in \R^{2d}
	. 
\end{align*}
\begin{lem}
	For $\varphi,\psi \in \Schwartz(\R^d)$ and $f \in \mathcal{C}^{\infty}_{\mathrm{pol} \, u}(\R^{2d}) \subseteq \Schwartz'(\R^{2d})$ we have 
	\begin{align*}
		\bscpro{\varphi}{\OpAel(f) \psi} = \frac{1}{(2\pi)^d} \int \dd X \, f(X) \, \WignerTrafoel(\varphi,\psi)(X) 
		. 
	\end{align*}
\end{lem}
\begin{proof}
	Since $f \in \mathcal{C}^{\infty}_{\mathrm{pol} \, u}(\R^{2d}) \subseteq \Schwartz'(\R^{2d})$, it is in the magnetic Moyal algebra $\mathcal{M}^B(\R^{2d})$ defined in \cite[Section~V.D.]{MantoiuPurice:magneticWeylCalculus:2004} and thus its quantization is a continuous operator $\Schwartz(\R^d) \longrightarrow \Schwartz(\R^d)$. Hence, the integral exists and we get the claim by direct computation. 
\end{proof}
The Wigner transform also leads to a `magnetic dequantization' -- once we know the operator kernel, we can reconstruct the distribution. We do not strive for full generality here. In particular, unless the operator has special properties, we cannot conclude that $f$ is in any Hörmander class. More sophisticated techniques are needed, \eg a Beals-type criterion \cite{MantoiuPurice:BealsCriterion:2008}. 
\begin{lem}\label{magBel:inverseWeylQ}
	Assume $B$ and $A$ satisfy Assumption~\ref{asympExp:assumption:usualConditionsBA} and $T \in \mathcal{B} \bigl ( L^2(\R^d) \bigr )$ is a bounded linear operator whose operator kernel $K_T$ is in $\Schwartz(\R^d \times \R^d)$. Then the \emph{inverse magnetic quantization} is given by 
	\begin{align}
		{\OpAel}^{-1}(T)(X) :=& \WignerTrafoel K_{T}(X) 
		%
		= \int \dd y \, e^{-i y \cdot \xi} \, e^{-i \lambda \GAe([\nicefrac{x}{\eps} - \nicefrac{y}{2} ,\nicefrac{x}{\eps} + \nicefrac{y}{2} ])} \,  K_{T} \bigl ( \tfrac{x}{\eps} - \tfrac{y}{2} , \tfrac{x}{\eps} + \tfrac{y}{2} \bigr ) 
		. 
	\end{align}
	This formula extends to operators with distributional kernels $K_T \in \Schwartz'(\R^d \times \R^d)$, \ie the kernels associated to continuous maps $\Schwartz(\R^d) \longrightarrow \Schwartz'(\R^d)$. 
\end{lem}
\begin{proof}
	If $T = \OpAel(f_T)$ is the magnetic quantization of $f_T \in \Schwartz(\R^{2d})$, then $\WignerTrafoel K_T = f_T \in \Schwartz(\R^{2d})$ follows from direct calculation, using the explicit form of the Wigner transform, Lemma~~\ref{magBel:WignerTransformExpectation}. Similarly, we confirm that $T = \OpAel \bigl ( \WignerTrafoel(K_T) \bigr )$ holds and $\WignerTrafoel K_T \in \Schwartz(\R^{2d})$ follows from $K_T \in \Schwartz(\R^d \times \R^d)$. 
	
	If the kernel of $T$ is a tempered distribution, then we can extend the formulas for $\OpAel$ and $\WignerTrafoel$ to $\Schwartz'(\R^{2d})$: Fourier transform, multiplication by a phase factor whose phase function is of tempered growth and a linear change of variables can all be extended to $\Schwartz'(\R^{2d})$ and thus it makes sense to write $\WignerTrafoel K_T$ after a suitable reinterpretation. Then $\WignerTrafoel K_T = f_T \in \Schwartz'(\R^{2d})$ is such that $\OpAel(f_T) = T : \Schwartz(\R^d) \longrightarrow \Schwartz'(\R^d)$. 
\end{proof}
%


\subsection{Equivalence of formulas for magnetic Weyl product} 
\label{asympExp:equivalenceProduct}

It turns out that the integral formula for the product found in \cite{MantoiuPurice:magneticWeylCalculus:2004,IftimieMantiouPurice:magneticPseudodifferentialOperators:2005} is not amenable to the derivation of an asymptotic expansion in $\eps$ and $\lambda$. Although an asymptotic expansion for $\eps = 1 = \lambda$ has been derived in \cite{IftimieMantiouPurice:magneticPseudodifferentialOperators:2005}, \emph{calculating} each term has proven to be very tedious and it is not obvious how to collect terms of the same power in $\eps$ and $\lambda$. Thus, we will use an equivalent formula for the magnetic Weyl product. From this, we derive closed formulas for the $(n,k)$ term by expanding the `twister' of the convolution in the next section. 
\begin{thm}[\cite{Mueller:productRuleGaugeInvariantWeylSymbols:1999, IftimieMantiouPurice:magneticPseudodifferentialOperators:2005}]\label{asympExp:thm:equivalenceProduct}
	Assume the magnetic field $B$ satisfies Assumption~\ref{asympExp:assumption:usualConditionsBA}. Then for two symbols $f \in \Hoerrd{m_1}$ and $g \in \Hoerrd{m_2}$, the magnetic Weyl product $f \magWel g$ is in symbol class $\Hoerrd{m_1 + m_2}$ and given by the oscillatory integral 
	\begin{align}
		(f \magWel g)(X) &= \frac{1}{(2 \pi)^{2d}} \int \dd Y \, \int \dd Z \, e^{+i \sigma(X,Y+Z)} \, e^{i \frac{\eps}{2} \, \sigma(Y,Z)} \, \cdot \notag \\
		&\qquad \qquad \cdot \OBel \bigl ( x-\tfrac{\eps}{2}(y+z),x + \tfrac{\eps}{2}(y-z) , x+\tfrac{\eps}{2}(y+z) \bigr ) \, \bigl ( \Fs^{-1} f \bigr )(Y) \, \bigl ( \Fs^{-1} g \bigr )(Z) \label{asympExp:FourierFormMagneticComposition} \\ 
		&= \frac{1}{(\pi\eps)^{2d}} \, \int \dd \tilde{Y} \, \int \dd \tilde{Z} \, e^{- i \frac{2}{\eps} \, \sigma(\tilde{Y} - X , \tilde{Z} - X)} \, \OBel \bigl ( x - \tilde{y} + \tilde{z} , -x + \tilde{y} + \tilde{z} , x + \tilde{y} + \tilde{z} \bigr ) \,  f(\tilde{Y}) \, g(\tilde{Z}) \notag 
	\end{align}
	where $\OBel(x,y,z) := e^{- i {\lambda}{\eps} \Gamma^B(\sexpval{x,y,z})}$ is the exponential of the magnetic flux through the triangle with corners $x$, $y$ and $z$. 
\end{thm}
\begin{proof}
	The Weyl product is defined implicitly by 
	\begin{align*}
		\OpAel(f) \, \OpAel(g) &=: \OpAel(f \magWel g) 
	\end{align*}
	and its quantization maps $\Schwartz(\R^d)$ to itself \cite[Proposition~21]{MantoiuPurice:magneticWeylCalculus:2004}. Combined with Theorem~\ref{magBel:inverseWeylQ}, this immediately yields 
	\begin{align*}
		(f \magWel g)(X) = \WignerTrafoel \bigl ( K_{\OpAel(f) \, \OpAel(g)} \bigr )(X) 
	\end{align*}
	where $K_{\OpAel(f) \, \OpAel(g)}$ is the kernel of $\OpAel(f) \, \OpAel(g)$. Here, we have chosen a vector potential $A$ which is associated to $B$ that also satisfies Assumption~\ref{asympExp:assumption:usualConditionsBA}. Although it is \emph{a priori} not clear that there must exist a \emph{symbol} $f \magWel g$, we will start with formal calculations and then use Corollary~\ref{appendix:existenceOscInt:remainder} to show that integral~\eqref{asympExp:FourierFormMagneticComposition} exists and is in the correct symbol class. 
	\medskip
	
	\noindent
	\textbf{Step 1: Rewrite in terms of Weyl system. }
	Plugging in the definition of $\OpAel$, we get 
	\begin{align*}
		\OpAel(f) \, \OpAel(g) 
		&= \frac{1}{(2\pi)^{2d}} \int \dd Y \, \int \dd Z \, \bigl ( \Fs^{-1} f \bigr )(Y) \, \bigl ( \Fs^{-1} g \bigr )(Z) \, \WeylSysAel(Y) \WeylSysAel(Z) \\
		&= \frac{1}{(2\pi)^{2d}} \int \dd Y \, \int \dd Z \, \bigl ( \Fs^{-1} f \bigr )(Y) \, \bigl ( \Fs^{-1} g \bigr )(Z) \;  e^{i \frac{\eps}{2} \, \sigma(Y,Z)} \cdot \\
		&\qquad \qquad \qquad \qquad \qquad \cdot \OBel(\Qe,\Qe+\eps y,\Qe+\eps y+\eps z) \WeylSysAel(Y+Z) \\
		&= \frac{1}{(2\pi)^{2d}} \int \dd Z \biggl ( \int \dd Y \, \bigl ( \Fs^{-1} f \bigr )(Y) \, \bigl ( \Fs^{-1} g \bigr )(Z-Y) \;  e^{i \frac{\eps}{2} \, \sigma(Y,Z)} \biggr . \cdot \\
		&\biggl . \qquad \qquad \qquad \qquad \qquad \cdot \OBel(\Qe,\Qe+\eps y,\Qe+\eps z) \biggr ) \WeylSysAel(Z) 
		. 
	\end{align*}
	In order to find the kernel of this operator, we need to find the kernel for $\hat{L}_{\eps,\lambda}(y,Z) := \OBel(\Qe,\Qe+\eps y,\Qe+\eps z) \WeylSysAel(Z)$ which parametrically depends on $y$ and $Z = (z,\zeta)$. 
	\medskip

	\noindent
	\textbf{Step 2: Find the operator kernel for $\hat{L}_{\eps,\lambda}(y,Z)$. } 
	Let $\varphi \in L^2(\R^d)$. Then we have 
	\begin{align*}
		\bigl ( \hat{L}_{\eps,\lambda}(y,Z) \varphi \bigr )(v) 
		&= \OBel(\eps v,\eps v+\eps y,\eps v+\eps z) \, e^{-i \eps (v + \nicefrac{z}{2}) \cdot \eta} \, e^{-i \lambda \GAe([v,v+z])} \, \varphi(v+z) \\
		&= \int \dd u \, e^{-i \eps (u - \nicefrac{z}{2}) \cdot \eta} \, e^{-i \lambda \GAe([u-z,u])} \,  \OBel(\eps u-\eps z,\eps u+\eps y-\eps z,\eps u) \, \delta \bigl ( u-(v+z) \bigr ) \, \varphi(u) \\
		&=: \int \dd u \, K_{L,\eps,\lambda}(y,Z;u,v) \, \varphi(u) 
		, 
	\end{align*}
	and we need to find $\WignerTrafoel K_{L,\eps,\lambda}(y,Z; \cdot , \cdot) (X)$, 
	\begin{align*}
		\WignerTrafoel K_{L,\eps,\lambda}(y,Z; \cdot , \cdot)(X) 
		&= \int \dd u \, e^{-i u \cdot \xi} \, e^{-i \lambda \GAe([\nicefrac{x}{\eps} - \nicefrac{u}{2} , \nicefrac{x}{\eps} + \nicefrac{u}{2}])} \, K_{L,\eps,\lambda} \bigl ( y,Z;\tfrac{x}{\eps} - \tfrac{u}{2},\tfrac{x}{\eps} + \tfrac{u}{2} \bigr ) \\
		&= e^{i \sigma(X,Z)} \, \OBel \bigl ( x - \tfrac{\eps}{2} z,x - \tfrac{\eps}{2} z + \eps y , x + \tfrac{\eps}{2} z \bigr ) =: L_{\eps,\lambda}(y,Z;X) 
		. 
	\end{align*}
	\medskip

	\noindent
	\textbf{Step 3: Magnetic composition law. }%
	Now we plug $L_{\eps,\lambda}(y,Z;X)$ back into the operator equation and obtain 
	\begin{align}
		(f \magWel g)(X) 
		&= \frac{1}{(2\pi)^{2d}} \int \dd Z \, \int \dd Y \, \bigl ( \Fs^{-1} f \bigr )(Y) \, \bigl ( \Fs^{-1} g \bigr )(Z-Y) \;  e^{i \frac{\eps}{2} \, \sigma(Y,Z)} \, L_{\eps,\lambda}(y,Z;X) 
		\notag \\ 
		&= \frac{1}{(2\pi)^{2d}} \int \dd Y \, \int \dd Z \, e^{i \sigma(X,Y+Z)} \,  e^{i \frac{\eps}{2} \, \sigma(Y,Z)} \, \OBel \bigl ( x - \tfrac{\eps}{2} (y+z),x + \tfrac{\eps}{2} (y-z) , x + \tfrac{\eps}{2} (y+z) \bigr ) 
		\cdot \notag \\
		&\qquad \qquad \qquad \qquad \qquad \cdot \bigl ( \Fs^{-1} f \bigr )(Y) \, \bigl ( \Fs^{-1} g \bigr )(Z)
		. 
		\label{asympExp:eqn:equivalence:eqn1}
	\end{align}
	This formula is the starting point for Müller's and our derivation of the asymptotic expansion of the product. However, we can show the equivalence to the product formula obtained by two of the authors in \cite{MantoiuPurice:magneticWeylCalculus:2004} by writing out the symplectic Fourier transforms, 
	\begin{align*}
		\mbox{RHS of \eqref{asympExp:eqn:equivalence:eqn1}}
		&= \frac{1}{(2\pi)^{4d}} \int \dd Y \, \int \dd \tilde{Y} \, \int \dd Z \, \int \dd \tilde{Z} \, e^{i \sigma(X - \tilde{Y},Y)} \, e^{i \sigma(X - \tilde{Z},Z)} \, e^{i \frac{\eps}{2} \, \sigma(Y,Z)} \, \cdot \\ 
		&\qquad \qquad \qquad \qquad \cdot \OBel \bigl ( x - \tfrac{\eps}{2} (y+z),x + \tfrac{\eps}{2} (y-z) , x + \tfrac{\eps}{2} (y+z) \bigr ) \, f(\tilde{Y}) \, g(\tilde{Z}) 
		. 
	\end{align*}
	If one writes out the exponential prefactors explicitly, sorts all terms containing $\xi$ and $\eta$ and then integrates over those variables, one obtains 
	\begin{align*}
		\frac{1}{(\pi\eps)^{2d}} \, \int \dd \tilde{Y} \, \int \dd \tilde{Z} \, e^{- i \frac{2}{\eps} \, \sigma(X - \tilde{Y} , X - \tilde{Z})} \, \OBel \bigl ( \tilde{y} - \tilde{z} + x , \tilde{y} + \tilde{z} - x , - \tilde{y} + \tilde{z} + x \bigr ) \,  f(\tilde{Y}) \, g(\tilde{Z}) 
		. 
	\end{align*}
	\medskip

	\noindent
	\textbf{Step 4: $f \magWel g \in \Hoerrd{m_1 + m_2}$. }
	The integral on the right-hand side of equation~\eqref{asympExp:eqn:equivalence:eqn1} satisfies the assumptions of Lemma~\ref{appendix:existenceOscInt:remainder} with $\tau = 1 = \tau'$ (keeping in mind that $\OBel$ satisfies the assumptions on $G_{\tau'}$ by Lemma~\ref{appendix:boundednessCoroBel}). Thus, the integral in equation~\eqref{asympExp:eqn:equivalence:eqn1} exists and is in symbol class $\Hoerrd{m_1 + m_2}$. 
\end{proof}
%


\subsection{Asymptotic expansion of the product} 
\label{asympExp:asymptoticExp}
To obtain an asymptotic expansion of the product, we adapt an idea by Folland to the present case \cite[p 108 f.]{Folland:harmonicAnalysisPhaseSpace:1989}: we expand the exponential of the \emph{twister} 
\begin{align*}
	e^{i \frac{\eps}{2} \sigma(Y,Z) - i \lambda \gamma^B_{\eps}(x,y,z)} &= e^{i T_{\eps,\lambda}(x,Y,Z)} \notag 
	\\
	&\asymp \sum_{n = 0}^{\infty} \sum_{k = 0}^{\infty} \eps^n \lambda^k \sum C_{n,k,a,\alpha,b,\beta}(x) \, y^{a} \eta^{\alpha} \, z^{b} \zeta^{\beta} \notag 
\end{align*}
as a polynomial in $y$, $\eta$, $z$ and $\zeta$ with coefficients $C_{n,k,a,\alpha,b,\beta} \in \mathcal{BC}^{\infty}(\R^d)$ that are bounded functions in $x$ with bounded derivatives to all orders. Then we can rewrite equation \eqref{asympExp:FourierFormMagneticComposition} as a convolution of derivatives of $f$ and $g$. Furthermore, we can show that there are always sufficiently many derivatives with respect to momenta so that each of the terms has the correct decay properties. 

The difficult part of the proof is to show the existence of certain oscillatory integrals. To clean up the presentation of the proof, we have moved these parts to Appendix \ref{appendix:existenceOscInt}. For simplicity, we also introduce the following nomenclature: 
\begin{defn}[Number of $q$s and $p$s]
	Let $B \in \mathcal{BC}^{\infty}(\R^d_x , \mathcal{C}^{\infty}_{\mathrm{pol}}(\R_Y^{2d} \times \R_Z^{2d}))$ be a function which can be decomposed into a finite sum of the form 
	\begin{align*}
		B(x,Y,Z) &= \sum_{\substack{\sabs{a} + \sabs{b} = n\\
		\sabs{\alpha} + \sabs{\beta} = k}} b_{a \alpha b \beta}(x,Y,Z) \, y^{a} \, \eta^{\alpha} \, z^{b} \, \zeta^{\beta} 
	\end{align*}
	where all $b_{a \alpha b \beta}$ smooth \emph{bounded} functions that depend on the multiindices $a , \alpha , b , \beta \in \N_0^d$. We then say that $B$ has $n$ $q$s (total number of factors in $y$ and $z$) and $k$ $p$s (total number of factors in $\eta$ and $\zeta$). 
\end{defn}
In the appendix we show how to convert $q$s into derivatives with respect to \emph{momentum} and $p$s into derivatives with respect to \emph{position}. Monomials of $x$ and $\xi$ multiplied with the symplectic Fourier transform of a Schwarz function $\varphi \in \Schwartz(\R^{2d})$ can be written as the symplectic Fourier transform of derivatives of $\varphi$ in $\xi$ and $x$: 
\begin{align*}
	x^{a} \xi^{\alpha} (\Fs \varphi)(X) &= \Fs \bigl ( (-i \partial_{{\xi}})^{a} (i \partial_{{x}})^{\alpha} \varphi \bigr ) (X)
\end{align*}
This manipulation can be made rigorous for symbols of Hörmander class $m$ with weight $\rho$. We see that derivatives with respect to momentum \emph{improve} decay by $\rho$ while those with respect to position do not alter the decay. In this sense, the decay properties of the integrals are determined by the number of $q$s and $p$s. 

Now we are in a position to prove the main result of this article, Theorem~\ref{intro:thm:main_result:asymptoticExpansion}: 
\begin{proof}[Theorem~\ref{intro:thm:main_result:asymptoticExpansion}]
	Let $\eprec \ll 1$. Then throughout the proof, we take $N \equiv N(\eprec,\eps,\lambda) \in \N_0$ to be as in the first part of Definition \ref{asympExp:defn:precision}, \ie $\eps^{N+1} < \eprec$ and $\lambda^{N+1} < \eprec$ hold. We will show that $f \magWel g - \sum_{n = 0}^N \sum_{k = 0}^n \eps^n \lambda^k \, (f \magWel g)_{(n,k)} = \ordereprec$. 
	\medskip
	
	\noindent
	\textbf{Step 1: Formal expansion of the twister. }
	We expand the exponential of the twister $e^{i \frac{\eps}{2} \sigma(Y,Z)} e^{- i \lambda \gBe(x,y,z)} = e^{i T_{\eps,\lambda}(x,Y,Z)}$ up to the $N$th term, 
	\begin{align*}
		e^{i T_{\eps,\lambda}(x,Y,Z)} &= \sum_{n = 0}^N \frac{i^n}{n!} \bigl ( T_{\eps,\lambda}(x,Y,Z) \bigr )^n + R_N(x,Y,Z) 
		. 
	\end{align*}
	The remainder 
	\begin{align}
		R_N(x,Y,Z) :&= \frac{1}{N!} \int_0^1 \dd \tau \, (1 - \tau)^N \partial_{\tau}^{N+1} e^{\tau u} \big \vert_{u = i T_{\eps,\lambda}(x,Y,Z)} \notag \\
		&= \frac{i^{N+1}}{N!} \bigl ( T_{\eps,\lambda}(x,Y,Z) \bigr )^{N+1} \int_0^1 \dd \tau \, (1 - \tau)^N \, e^{i \tau T_{\eps,\lambda}(x,Y,Z)} \label{asympExp:eqn:remainder1}
	\end{align}
	is treated in Step 3, right now we are only concerned with the first term. If we plug in the asymptotic expansion of the flux $\gBe$ derived in Lemma~\ref{app:lemmaFormalExpansionTwister} up to $N'$th order with $N' \geq N$, then we obtain 
	\begin{align}
		\bigl ( T_{\eps,\lambda}(x,Y,Z) \bigr )^n &= 
		\Bigl ( \tfrac{\eps}{2} \sigma(Y,Z) + \lambda \mbox{$\sum_{n' = 1}^{N'}$} \eps^{n'} \mathcal{L}_{n'}(x,y,z) + \lambda R_{N'}[\gBe](x,y,z) \Bigr )^n \notag \\
		&= \sum_{l = 0}^n \noverk{n}{l} \Bigl ( \tfrac{\eps}{2} \sigma(Y,Z) + \lambda \mbox{$\sum_{n' = 1}^{N'}$} \eps^{n'} \mathcal{L}_{n'}(x,y,z) \Bigr )^{n - l} \, \bigl ( \lambda R_{N'}[\gBe](x,y,z) \bigr )^l \notag \\
		&=: \Bigl ( \tfrac{\eps}{2} \sigma(Y,Z) + \lambda \mbox{$\sum_{n' = 1}^{N'}$} \eps^{n'} \mathcal{L}_{n'}(x,y,z) \Bigr )^{n} + R_{N' \, n}[T_{\eps,\lambda}](x,Y,Z) 
		. 
		\label{asympExp:eqn:remainder2}
	\end{align}
	Again, we focus on the first term of the expansion and treat the remainder separately in Step 3: 
	\begin{align*}
		\Bigl ( \tfrac{\eps}{2} \sigma(Y,Z) + &\lambda \mbox{$\sum_{n' = 1}^{N'}$} \eps^{n'} \mathcal{L}_{n'}(x,y,z) \Bigr )^n = 
		\\
		&= \sum_{k = 0}^n \sum_{\sum_{j = 1}^{N'} k_j = k} \eps^{(n - k) + \sum_{j = 1}^{N'} j k_j} \lambda^k  \, \frac{n!}{(n - k)! \, k_1 ! \cdots k_{N'}!} \Bigl ( \tfrac{1}{2} \sigma(Y,Z) \Bigr )^{n-k} \, \prod_{j = 1}^{N'} {\mathcal{L}_j}^{k_j}(x,y,z) 
	\end{align*}
	Now we define $\mathcal{L}_0(Y,Z) := \tfrac{1}{2} \sigma(Y,Z)$ to clean up the presentation, include the sum over $n$ again and sort by powers of $\eps$ and $\lambda$, 
	\begin{align*}
		\sum_{n = 0}^N &\frac{i^n}{n!} \Bigl ( \tfrac{\eps}{2} \sigma(Y,Z) + \lambda \mbox{$\sum_{n' = 1}^{N'}$} \eps^{n'} \mathcal{L}_{n'}(x,y,z) \Bigr )^n 
		= \\
		&= \sum_{n = 0}^N \frac{i^n}{n!} \sum_{\sum_{j = 0}^{N'} k_j = n} \eps^{k_0 + \sum_{j = 1}^{N'} j k_j} \lambda^{n - k_0}  \, \frac{n!}{k_0! \, k_1 ! \cdots k_{N'}!} {\mathcal{L}_0}^{k_0}(Y,Z) \prod_{j = 1}^{N'} {\mathcal{L}_j}^{k_j}(x,Y,Z) \\
		&= \sum_{n = 0}^{N \, N'} \sum_{k = 0}^n \eps^{n} \lambda^k 
		\sum_{\substack{k_0 + \sum_{j = 1}^{N'} j k_j = n \\
		\sum_{j=1}^{N'} k_j = k}} \frac{i^{k + k_0}}{k_0! \, k_1 ! \cdots k_{N'}!} \, {\mathcal{L}_0}^{k_0}(Y,Z) \prod_{j = 1}^{N'} {\mathcal{L}_j}^{k_j}(x,Y,Z) 
		=: \sum_{n = 0}^{N \, N'} \sum_{k = 0}^n \eps^{n} \lambda^k \mathcal{T}_{n,k}(x,Y,Z) 
		. 
	\end{align*}
	\textbf{Step 2: Existence of the $(n,k)$ term. }
	The properties of the $(n,k)$th term of the product 
	\begin{align}
		(f \magWel g)_{(n,k)}(X) &= \frac{1}{(2\pi)^{2d}} \int \dd Y \int \dd Z \, e^{i \sigma(X,Y+Z)} \, \mathcal{T}_{n,k}(x,Y,Z) \, (\Fs^{-1} f)(Y) \, (\Fs^{-1} g)(Z) 
		\label{magBel:integralnkTerm}
	\end{align}
	can be deduced from the properties of $\mathcal{T}_{n,k}$: we proceed by showing that $\mathcal{T}_{n,k}$ is a polynomial with $x$-dependent prefactors that contains $n+k$ $q$s (powers of $y$ and $z$) and \emph{at most} $n-k$ $p$s (powers of $\eta$ and $\zeta$). $\mathcal{L}_0$ is the non-magnetic symplectic form and contains $1$ $q$ and $1$ $p$. Hence, the $k_0$th power of $\mathcal{L}_0$ contributes $k_0$ $q$s and an equal amount of $p$s. By Lemma~\ref{app:lemmaFormalExpansionTwister}, the magnetic terms $\mathcal{L}_j$, $j \geq 1$, contribute $j+1$ $q$s and no $p$s. In this sense, magnetic terms improve decay. By conditions imposed on the indices appearing in the definition of $\mathcal{T}_{n,k}$, we deduce there are 
	\begin{align*}
		k_0 + \sum_{j = 1}^{N'} (j + 1) k_j = k_0 + \sum_{j = 1}^{N'} j k_j + \sum_{j = 1}^{N'} k_j = n + k
	\end{align*}
	$q$s and $k_0$ $p$s. As $0 \leq k_0 \leq n-k$, Lemma~\ref{appendix:existenceOscInt:Lemma2} implies the existence of integral \eqref{magBel:integralnkTerm} and that it belongs to the correct symbol class, namely $(f \magWel g)_{(n,k)} \in \Hoerrd{m_1 + m_2 - (n+k) \rho}$. 
	\medskip
	
	\noindent
	\textbf{Step 3: Existence of remainders. }
	There are two remainders we need to control, equations~\eqref{asympExp:eqn:remainder1} and \eqref{asympExp:eqn:remainder2}: the first one stems from the Taylor expansion of the exponential, the second one has its origins in the expansion of the magnetic flux, 
	\begin{align*}
		R_N^{\Sigma}(x,Y,Z) := R_{N}(x,Y,Z) + \sum_{n = 1}^N \frac{i^n}{n!} R_{N' \, n}[T_{\eps,\lambda}](x,Y,Z) 
		. 
	\end{align*}
	The remainder of the product is obtained after integration, 
	\begin{align}
		\tilde{R}_N(X) := \frac{1}{(2\pi)^{2d}} \int \dd Y \int \dd Z \, e^{i \sigma(X,Y+Z)} \, R_N^{\Sigma}(x,Y,Z) \, (\Fs^{-1} f)(Y) \, (\Fs^{-1} g)(Z) 
		. 
		\label{asympExp:eqn:total_remainder}
	\end{align}
	We have to show that (i) the integral exists, (ii) it is in the correct symbol class and (iii) it is of the right order in $\eps$ and $\lambda$. Points (i) and (ii) are the content of Lemma~\ref{appendix:existenceOscInt:remainder} and we have to show that each of the two contributions to the remainder satisfies the assumptions. 
	\medskip
	
	\noindent
	The first contribution to $\tilde{R}_N$ stems from the Taylor expansion of the exponential, 
	\begin{align*}
		\frac{1}{(2\pi)^{2d}} \int \dd Y &\int \dd Z \, e^{i \sigma(X,Y+Z)} \, 
			\frac{1}{N!} \int_0^1 \dd \tau \, (1 - \tau)^N \partial_{\tau}^{N+1} e^{\tau u} \big \vert_{u = i T_{\eps,\lambda}(x,Y,Z)}
			\, (\Fs^{-1} f)(Y) \, (\Fs^{-1} g)(Z) = \\
		&= \frac{1}{(2\pi)^{2d}} \int_0^1 \dd \tau \, (1 - \tau)^N \, \int \dd Y \int \dd Z \, e^{i \sigma(X,Y+Z)} \, 
			\frac{i^{N+1}}{N!} \bigl ( T_{\eps,\lambda}(x,Y,Z) \bigr )^{N+1} e^{-i \tau \lambda \gBe(x,y,z)} \cdot 
			\\
			&\qquad \qquad \qquad \qquad \cdot e^{i \tau \frac{\eps}{2} \sigma(Y,Z)}
			\, (\Fs^{-1} f)(Y) \, (\Fs^{-1} g)(Z) 
			. 
	\end{align*}
	The first factor, $\bigl ( T_{\eps,\lambda}(x,Y,Z) \bigr )^{N+1}$, can be expanded in powers of $\sigma(Y,Z)$ and $\gBe(x,y,z)$: 
	\begin{align*}
		\bigl ( T_{\eps,\lambda}(x,Y,Z) \bigr )^{N+1} &= \eps^{N+1} \sum_{l = 0}^{N+1} \noverk{N+1}{l} \lambda^l \, \bigl ( \tfrac{1}{2} \sigma(Y,Z) \bigr )^{N+1-l} \bigl ( \underbrace{\tfrac{1}{\eps} \gBe(x,y,z)}_{= \orderone} \bigr )^l  
	\end{align*}
	As $\eps^{N+1} < \eprec$ holds by definition of $N$, the first term of the remainder is of the correct order. The decay properties are dominated by $\bigl ( \sigma(Y,Z) \bigr )^{N+1}$ with $N+1$ $p$s and $N+1$ $q$s. All other terms contribute less than $N+1$ $p$s and more than $N+1$ $q$s since by Lemma~\ref{app:lemmaFormalExpansionTwister}, $\gBe$ is of order $\eps$ and contributes 2 $q$s and no $p$s. Furthermore, Lemma~\ref{appendix:boundednessLemmagBel} gives polynomial bounds of derivatives of $\gBe$: 
	\begin{align*}
		\babs{\partial_x^{a} \gBe(x,y,z)} \leq C_{a} \, \bigl ( \sexpval{y} + \sexpval{z} \bigr )^{\sabs{a}} 
	\end{align*}
	A similar bound holds for the exponential of the flux (Corollary~\ref{appendix:boundednessCoroBel}):
	\begin{align*}
		\babs{\partial_x^{a} e^{- i \lambda \gBe(x,y,z)}(x,y,z)} \leq C_{a} \, \sexpval{y}^{\sabs{a}} \sexpval{z}^{\sabs{a}} 
		&&
		\forall a \in \N_0^d 
	\end{align*}
	Altogether, $\bigl ( T_{\eps,\lambda}(x,Y,Z) \bigr )^{N+1} \, e^{-i \tau \lambda \gBe(x,y,z)}$ satisfies the conditions on $G_{\tau'}$ in Lemma~\ref{appendix:existenceOscInt:remainder} (with $\tau = \tau'$) which implies 
	\begin{align*}
		\frac{1}{(2\pi)^{2d}} \int_0^1 \dd \tau \, (1 - \tau)^N \, \int \dd Y \int \dd Z \, &e^{i \sigma(X,Y+Z)} \, 
			\frac{i^{N+1}}{N!} \bigl ( T_{\eps,\lambda}(x,Y,Z) \bigr )^{N+1} e^{-i \tau \lambda \gBe(x,y,z)} \cdot 
			\\
			&\cdot e^{i \tau \frac{\eps}{2} \sigma(Y,Z)}
			\, (\Fs^{-1} f)(Y) \, (\Fs^{-1} g)(Z) 
	\end{align*}
	exists as an oscillatory integral and belongs to symbol class $\Hoerrd{m_1 + m_2 - (N+1)\rho}$. 
	\medskip
	
	\noindent
	The second contribution which involves 
	\begin{align*}
		R_{N' \, n}[T_{\eps,\lambda}](x,Y,Z) &= \sum_{l = 1}^n \noverk{n}{l} \Bigl ( \tfrac{\eps}{2} \sigma(Y,Z) + \lambda \mbox{$\sum_{n' = 1}^{N'}$} \eps^{n'} \mathcal{L}_{n'}(x,y,z) \Bigr )^{n - l} \, \bigl ( \lambda R_{N'}[\gBe](x,y,z) \bigr )^l
	\end{align*}
	can be estimated analogously: by Lemma~\ref{app:lemmaFormalExpansionTwister}, $R_{N'}[\gBe]$ is of order $\ordere{N'+1}$ (the largest prefactor is $\eps^{N'+1} \lambda < \eprec$) and contains $N' + 2$ $q$s. So the terms in the above sum contain at least $N' + 1 \geq N + 1$ more $q$s than $p$s and another application of Lemma~\ref{appendix:existenceOscInt:remainder} (with $\tau = 0$) implies that the second contribution to $\tilde{R}_N$ exists as an oscillator integral and is of symbol class $\Hoerrd{m_1 + m_2 - (N'+1) \rho} \subseteq \Hoerrd{m_1 + m_2 - (N+1) \rho}$. 
	\medskip
	
	\noindent
	Altogether, we conclude that $\tilde{R}_N$ exists pointwise, is of symbol class $\Hoerrd{m_1 + m_2 - (N+1) \rho}$ as long as $N' \geq N$ and hence $f \magWel g - \sum_{n = 0}^N \sum_{k = 0}^n \eps^n \lambda^k \, (f \magWel g)_{(n,k)}  = \ordereprec$. This concludes the proof.
\end{proof}
If we do not have a separation of spatial scales, \ie $\eps = 1$, but weak coupling to the magnetic field, we can still expand the product $\magWel$ as a power series in $\lambda$. This is also the starting point of the $\lambda$-$\eps$ expansion which coincides with the $\eps$-$\lambda$ expansion. 
\begin{thm}\label{asympExp:thm:lambdaExpansion}
	Assume the magnetic field $B$ satisfies Assumption~\ref{asympExp:assumption:usualConditionsBA}; then for $\lambda \ll 1$ and $\eps \leq 1$, we can expand the $\lambda$ Weyl product of $f \in \Hoerrd{m_1}$ and $g \in \Hoerrd{m_2}$ asymptotically in $\lambda$ such that 
	\begin{align*}
		f \magWel g - \sum_{k = 0}^N \lambda^k (f \magWel g)_{(k)} \in \Hoerrd{m_1 + m_2 - 2(N+1) \rho} 
	\end{align*}
	where $(f \magWel g)_{(k)} \in \Hoerrd{m_1 + m_2 - 2k \rho}$ is given by equation~\eqref{asympExp:eqn:lambdaExpansion:kth_term}. In particular, the zeroth-order term reduces to the \emph{non-magnetic} Weyl product, $(f \magWel g)_{(0)} = f \Weyle g$. We have explicit control over the remainder (equation~\eqref{aysmpExp:eqn:lambda_expansion:remainder}): if we expand the product up to $N$th order in $\lambda$, the remainder is of order $\orderl{N+1}$ and in symbol class $\Hoerrd{m_1 + m_2 - 2(N+1)\rho}$. 
\end{thm}
\begin{proof}
	Assume we want to expand up to $N$th order in $\lambda$. We will show $f \magWel g - \sum_{k = 0}^N \lambda^k (f \magWel g)_{(k)} = \orderl{N+1}$ and that the difference is in $\Hoerrd{m_1 + m_2 - 2(N+1) \rho}$. 
	\medskip
	
	\noindent
	\textbf{Step 1: Expansion of exponential flux. }
	If $\eps$ is not necessarily small, we cannot expand the magnetic flux integral $\gBe$ in powers of $\eps$ anymore. However, we will keep $\eps$ as a \emph{bookkeeping device}. Expanding the exponential of the magnetic flux, we get 
	\begin{align*}
		e^{i T_{\eps,\lambda}(x,Y,Z)} &= e^{i \frac{\eps}{2} \sigma(Y,Z)} e^{- i \lambda \gBe(x,y,z)} \\
		&= e^{i \frac{\eps}{2} \sigma(Y,Z)} \Bigl ( \mbox{$\sum_{k = 0}^N$} \lambda^k \tfrac{(-i)^k}{k!} \bigl ( \gBe(x,y,z) \bigr )^k + R_N(x,y,z)  \Bigr ) 
		. 
	\end{align*}
	The remainder is of order $\lambda^{N+1}$ and has $2(N+1)$ $q$s, 
	\begin{align*}
		R_N(x,y,z) &= \frac{1}{N!} \bigl ( -i \lambda \gBe(x,y,z) \bigr )^{N+1} \int_0^1 \dd \tau' \, (1 - \tau')^N \, e^{-i \lambda \tau' \gBe(x,y,z)} 
		. 
	\end{align*}
	This can be seen more readily once we define $- \eps \Bte_{lj}(x,y,z) \, y_l z_j := \gBe(x,y,z)$ to emphasize that $\gBe$ contains $\eps$ as a prefactor and $2$ $q$s. Using the antisymmetry of $B_{lj}$, there is a simple explicit expression for $\Bte_{lj}$ (see proof of Lemma~\ref{app:lemmaFormalExpansionTwister}): 
	\begin{align*}
		\Bte_{lj}(x,y,z) 
		= \frac{1}{2} \int_{-\nicefrac{1}{2}}^{+\nicefrac{1}{2}} \dd t \int_0^1 \dd s \, s \, \bigl [ B_{lj} \bigl ( x + \eps s (t y - \nicefrac{z}{2}) \bigr )  + B_{lj} \bigl ( x + \eps s (\nicefrac{y}{2} + t z) \bigr ) \bigr ] 
		= \orderone 
	\end{align*}
	\textbf{Step 2: Existence of $k$th-order term. }
	Then the expansion can be rewritten so that we can separate off factors of $y$, $z$ and $\eps$. The $k$th order term contains $2k$ $q$s and no $p$s,  
	\begin{align*}
		\frac{(-i)^k}{k!} \bigl ( \gBe(x,y,z) \bigr )^k &= \eps^k \, \frac{i^k}{k!} \prod_{m = 1}^k \Bte_{l_m j_m}(x,y,z) \, y_{l_m} z_{j_m} 
		. 
	\end{align*}
	By Lemma~\ref{appendix:existenceOscInt:remainder} (with $\tau = 1 = \tau'$) the $k$th order term of the product 
	\begin{align}
		(f \magWel g)_{(k)}(X) :&= \frac{\eps^k}{(2\pi)^{2d}} \int \dd Y \int \dd Z \, e^{i \sigma(X,Y+Z)} \, e^{i \frac{\eps}{2} \sigma(Y,Z)} \, \biggl ( \frac{i^k}{k!} \prod_{m = 1}^k \Bte_{l_m j_m}(x,y,z) \, y_{l_m} z_{j_m}  \biggr ) 
		\cdot \notag \\
		&\qquad \qquad \qquad \cdot (\Fs^{-1} f)(Y) \, (\Fs^{-1} g)(Z) 
		\notag \\
		&= \frac{\eps^k}{(2\pi)^{2d}} \int \dd Y \int \dd Z \, e^{i \sigma(X,Y+Z)} \, e^{i \frac{\eps}{2} \sigma(Y,Z)} \, \biggl ( \frac{i^{3k}}{k!} \prod_{m = 1}^k \Bte_{l_m j_m}(x,y,z) \biggr ) 
		\cdot \notag \\
		&\qquad \qquad \qquad \cdot \bigl ( \Fs^{-1} (\partial_{\tilde{\eta}_{j_1}} \cdots \partial_{\tilde{\eta}_{j_k}} f) \bigr )(Y) \, \bigl ( \Fs^{-1} (\partial_{\tilde{\zeta}_{j_1}} \cdots \partial_{\tilde{\zeta}_{j_k}} g) \bigr )(Z) 
		\label{asympExp:eqn:lambdaExpansion:kth_term}
	\end{align}
	exists and is of symbol class $\Hoerrd{m_1 + m_2 - 2k \rho}$. 
	\medskip

	\noindent
	\textbf{Step 3: Existence of remainder. }
	The remainder is of order $\lambda^{N+1}$ and has $2(N+1)$ $q$s. It contains $\eps^{N+1}$ as a prefactor as well which will be of importance in the proof of the next theorem. By Lemma~\ref{appendix:boundednessLemmagBel} and Corollary~\ref{appendix:boundednessCoroBel}, the integral in $R_N$ over the exponential of the magnetic flux is bounded and its derivatives can be bounded polynomially in $y$ and $z$, 
	\begin{align*}
		R_N(x,y,z) &= \lambda^{N+1} \frac{\eps^{N+1}}{N!} \bigl ( \Bte_{lj}(x,y,z) \, y_l z_j \bigr )^{N+1} \int_0^1 \dd \tau' \, (1 - \tau')^N \, e^{-i \lambda \tau' \gBe(x,y,z)} 
		. 
	\end{align*}
	This means $R_N$ satisfies the conditions on $G_{\tau'}$ in Lemma~\ref{appendix:existenceOscInt:remainder} (with $\tau = 1$) and we conclude that 
	\begin{align}
		\tilde{R}_N(X) := \frac{1}{(2\pi)^{2d}} \int \dd Y \int \dd Z \, e^{i \sigma(X,Y+Z)} \, e^{i \frac{\eps}{2} \sigma(Y,Z)} \, R_N(x,y,z) \, (\Fs^{-1} f)(Y) \, (\Fs^{-1} g)(Z) 
		\label{aysmpExp:eqn:lambda_expansion:remainder}
	\end{align}
	exists and is in symbol class $\Hoerrd{m_1 + m_2 - 2(N+1) \rho}$. 
\end{proof}
The next statement is central to this paper, because it tells us we can speak of \emph{the} two-parameter expansion of the product. 
\begin{thm}\label{asympExp:thm:lambdaEpsExpansion}
	Assume that the magnetic field $B$ satisfies Assumption~\ref{asympExp:assumption:usualConditionsBA} and $\eps \ll 1$ in addition to $\lambda \ll 1$. Then we can expand each term of the $\lambda$ expansion of $f \magWel g$ in $\eps$, $f \in \Hoerrd{m_1}$, $g \in \Hoerrd{m_2}$, and obtain the same as in Theorem~\ref{intro:thm:main_result:asymptoticExpansion}. Hence we can speak of \emph{the} two-parameter expansion of the product $\magWel$. 
\end{thm}
\begin{proof}
	\textbf{Step 1: Precision of expansion. }
	Assume we have expanded the magnetic product $\magWel$ up to $N_0$th power in $\lambda$. Then for the remainder of the proof, we fix $N \equiv N(\eprec,\eps,\lambda) \in \N_0$ as in Definition~\ref{asympExp:defn:precision} for $\eprec = \lambda^{N_0}$. 
	\medskip

	\noindent
	\textbf{Step 2: Equality of $(n,k)$ terms of expansion. }
	Now to the expansion itself. The two terms we need to expand are the non-magnetic twister $e^{i \frac{\eps}{2} \sigma(Y,Z)}$ and the $k$th power of the magnetic flux integral $\gBe$ in $\eps \ll 1$: we choose $N' , N'' \geq N$ and write the $k$th order of the $\lambda$ expansion as 
	\begin{align*}
		(f \magBl g)_{(k)}(X) &= \frac{1}{(2\pi)^{2d}} \int \dd Y \int \dd Z \, e^{i \sigma(X,Y+Z)} \, e^{i \frac{\eps}{2} \sigma(Y,Z)} \, \frac{(-i)^k}{k!} \bigl ( \gBe(x,y,z) \bigr )^k \, (\Fs^{-1} f)(Y) \, (\Fs^{-1} g)(Z) 
		\\ 
		&= \frac{1}{(2\pi)^{2d}} \int \dd Y \int \dd Z \, e^{i \sigma(X,Y+Z)} \, \Bigl ( \mbox{$\sum_{n = 0}^{N'}$} \eps^n \tfrac{i^n}{n!} \bigl ( \tfrac{1}{2} \sigma(Y,Z) \bigr )^n + R_{N'}[\sigma](Y,Z) \Bigr ) \cdot \\
		&\qquad \qquad \cdot \frac{(-i)^k}{k!} 
	\Bigl ( \bigl ( \mbox{$\sum_{j = 1}^{N''} \eps^j \mathcal{L}_j(x,y,z)$} \bigr )^k + R_{N'' \, k}[\mathcal{L}R](x,y,z) \Bigr ) 
		\, (\Fs^{-1} f)(Y) \, (\Fs^{-1} g)(Z)
		. 
	\end{align*}
	The remainders are given explicitly in Step 3, equations~\eqref{asympExp:eqn:proof_equivalence_expansions:RNsigma} and \eqref{asympExp:eqn:proof_equivalence_expansions:RNgBe}. The $(n,k)$ terms of the expansion originate from the first of these terms, \ie we need to look at 
	\begin{align*}
		\sum_{n = 0}^{N'} &\eps^n \frac{i^n}{n!} \Bigl ( \tfrac{1}{2} \sigma(Y,Z) \Bigr )^n \, \Bigl ( \mbox{$\sum_{j = 1}^{N''} \eps^j \mathcal{L}_j(x,y,z)$} \Bigr )^k = 
		\\
		&= \sum_{n = 0}^{N'} \sum_{\sum_{j = 1}^{N''} k_j = k} \eps^{n + \sum_{j = 1}^{N''} j k_j} \frac{i^{n+k}}{n! k_1! \cdots k_{N''}!} 
		\bigl ( \tfrac{1}{2} \sigma(Y,Z) \bigr )^n \, \prod_{j = 1}^{N''} {\mathcal{L}_j}^{k_j}(x,y,z) 
	\end{align*}
	to obtain the $(n,k)$ term of this expansion. The remaining three terms define the remainder which will be treated in the last step. We define $\mathcal{L}_0(Y,Z) := \tfrac{1}{2} \sigma(Y,Z)$, $k_0 := n$ and recognize the result from Theorem~\ref{intro:thm:main_result:asymptoticExpansion}, the terms match: 
	\begin{align*}
		\sum_{n = k}^{N' \, N''}& \sum_{\substack{k_0 + \sum_{j = 1}^{N''} j k_j = n\\
		\sum_{j = 1}^{N''} k_j = k}} \eps^{n} \frac{i^{k+k_0}}{k_0! k_1! \cdots k_{N''}!} 
		{\mathcal{L}_0}^{k_0}(Y,Z) \, \prod_{j = 1}^{N''} {\mathcal{L}_j}^{k_j}(x,y,z)
	\end{align*}
	Obviously, the arguments made in the proof of Theorem~\ref{intro:thm:main_result:asymptoticExpansion} can be applied here as well, and we conclude that the $(n,k)$ term exists and is in the correct symbol class, $\Hoerrd{m_1 + m_2 - (n+k)\rho}$. 
	\medskip
	
	\noindent
	\textbf{Step 3: Existence of remainders. }
	The remainders of the expansions of $e^{i \frac{\eps}{2} \sigma(Y,Z)}$ and $\bigl ( \gBe(x,y,z) \bigr )^k$, 
	\begin{align}
		R_{N'}[\sigma](Y,Z) &= \eps^{N'+1} \, \frac{i^{N'+1}}{N'!} \bigl ( \tfrac{1}{2} \sigma(Y,Z) \bigr )^{N'+1} \, \int_0^1 \dd \tau ( 1 - \tau )^{N'} e^{i \tfrac{\eps}{2} \tau \sigma(Y,Z)} 
		\label{asympExp:eqn:proof_equivalence_expansions:RNsigma}
	\end{align}
	and 
	\begin{align}
		R_{N'' \, k}[\mathcal{L}R](x,y,z) &= \sum_{l = 1}^k \noverk{k}{l} \bigl ( \mbox{$\sum_{j = 1}^{N''} \eps^j \mathcal{L}_j(x,y,z)$} \bigr )^{k - l} \, \bigl ( R_{N''}[\gBe](x,y,z) \bigr )^l
		\label{asympExp:eqn:proof_equivalence_expansions:RNgBe}
	\end{align}
	with $R_{N''}[\gBe](x,y,z)$ as in Lemma~\ref{app:lemmaFormalExpansionTwister}, lead to three terms in the total remainder: 
	\begin{align*}
		R_{N N' N'' \, k}^{\Sigma}(x,Y,Z) &= R_{N'}[\sigma](Y,Z) \Bigl ( \bigl ( \mbox{$\sum_{j = 1}^{N''} \eps^j \mathcal{L}_j(x,y,z)$} \bigr )^k 
		+ R_{N'' \, k}[\mathcal{L}R](x,y,z) \Bigr ) 
			+ \\
			&\qquad 
		+ \Bigl ( \mbox{$\sum_{n = 0}^{N'}$} \eps^n \tfrac{i^n}{n!} \bigl ( \tfrac{1}{2} \sigma(Y,Z) \bigr )^n \Bigr ) \, R_{N'' \, k}[\mathcal{L}R](x,y,z) 
	\end{align*}
	Going through the motions of the proof to Theorem~\ref{intro:thm:main_result:asymptoticExpansion},  we count $p$s and $q$s, and then apply Lemma~\ref{appendix:existenceOscInt:remainder}. The first remainder, $R_{N'}[\sigma](Y,Z)$, is of order $\eps^{N'+1} < \eprec$ in $\eps$ and contributes $N'+1$ $q$s and $p$s. By Lemma~\ref{app:lemmaFormalExpansionTwister}, $R_{N''}[\gBe]$ contributes at least $N''+2$ $q$s and all prefactors are less than or equal to $\eps^{N''+1} < \eprec$. Thus the terms in $R_{N'' \, k}[\mathcal{L}R]$ contain at least $N'' + 2$ $q$s (for all $k \leq N$) and prefactors that are at most $\eps^{N''+1} < \eprec$. Hence, the total remainder exists as an oscillatory integral, is $\ordereprec$ small and in symbol class $\Hoerrd{m_1 + m_2 - (N+1)\rho}$. 
\end{proof}
\begin{remark}
	The asymptotic expansion of $\magWel$ can be immediately extended to an expansion of products of \emph{semiclassical two-parameter symbols} (see Definition~\ref{asympExp:defn:2ParameterSemiSymbol}). 
\end{remark}
%


\subsection{Relation between magnetic and ordinary Weyl calculus} 
\label{asympExp:magWQminSub}

In a previous work \cite{IftimieMantiouPurice:magneticPseudodifferentialOperators:2005}, Iftimie et al have investigated the relation between magnetic Weyl quantization and regular Weyl quantization combined with minimal substitution, the `usual' recipe to couple a quantum system to a magnetic field. However, since there were no small parameters $\eps$ and $\lambda$, we have to revisit their statements and adapt them to the present case. 

Let us define $\minsAl(X) := \xi - \lambda A(x)$ as coordinate transformation which relates momentum and kinetic momentum. With a little abuse of notation, we will also use $f \circ \minsAl(X) := f(x,\minsAl(X))$ to transform functions. In general, $\OpAel(f) \neq \Ope(f \circ \minsAl)$ since the latter is \emph{not} manifestly covariant. However, we would like to be able to compare results obtained with magnetic Weyl calculus to those obtained with usual Weyl calculus and minimal substitution. To show how the two calculi are connected, we need to make slightly stronger assumptions on the magnetic \emph{vector potential}. This may appear contrary to the spirit of the rest of the paper where it has been emphasized that restrictions should be placed on the magnetic \emph{field}. The necessity arises, because usual, non-magnetic Weyl calculus is used in this section. 
\begin{assumption}\label{asympExp:magWQminSub:assumption:stricterBA}
	We assume that the magnetic field $B$ is such that we can find a vector potential $A$ whose components satisfy 
	\begin{align*}
		\babs{\partial_x^{a} A_l(x)} \leq C_{a} 
		, 
		&& \forall 1 \leq l \leq d, \, \sabs{a} \geq 1, \, a \in \N_0^d 
		. 
	\end{align*}
\end{assumption}
In particular, this implies that the magnetic field $B = \dd A$ satisfies Assumption~\ref{asympExp:assumption:usualConditionsBA}, \ie its components are  $\mathcal{BC}^{\infty}$~functions. It is conceptually useful to introduce the line integral 
\begin{align}
	\Gamma^A(x,y) := \int_0^1 \dd s \, A \bigl (x + s(y - x) \bigr )
\end{align}
which is related to the circulation $\Gamma^A([x,y]) = (y - x) \cdot \Gamma^A(x,y)$; similarly, $\GAe([x,y]) =: (y - x) \cdot \GAe(x,y)$ defines the scaled line integral. This allows us to rewrite the integral kernel of a magnetic pseudodifferential operator $\OpAel(f)$ for $f \in \Hoerrd{m}$ as 
\begin{align}
	K_{f,\eps,\lambda}(x,y) = \int \dd \eta \, e^{- i y \cdot \eta} f \bigl ( \tfrac{\eps}{2} (x+y) , \eta - \lambda \GAe(x,y) \bigr ) . \label{asympExp:magWQminSub:magneticIntKernelRewritten}
\end{align}
If we had used minimal substitution instead, then we would have to replace the line integral $\GAe(x,y)$ by its mid-point value $A \bigl ( \tfrac{\eps}{2} (x+y) )$. 

\begin{thm}[\cite{IftimieMantiouPurice:magneticPseudodifferentialOperators:2005}]\label{asympExp:thm:magWQusualWQ}
	Assume the magnetic field satisfies Assumption~\ref{asympExp:magWQminSub:assumption:stricterBA}. Then for any $f \in \Hoerrd{m}$ there exists a \emph{unique} $g \in \Hoerrd{m}$ such that $\OpAel(f) = \Ope(g \circ \minsAl)$. $g$ can be expressed as an asymptotic series $g \asymp \sum_{n = 0}^{\infty} \sum_{k = 1}^n \eps^n \lambda^k \, g_{n,k}$, where $g_{n,k} \in \Hoerrd{m - (n+k) \rho}$ for all $n \geq 1$, and 
	\begin{align}
		\sum_{k = 1}^n \lambda^k g_{n,k}(x,\xi) = \eps^{-n} \sum_{\sabs{a} = n} \frac{1}{a !} \, \bigl ( i \partial_y \bigr )^{a} \Bigl ( \partial_{\xi}^{a} f \bigl ( x , \xi - \lambda \Gamma^A(x + \tfrac{\eps}{2} \, y , x - \tfrac{\eps}{2} \, y) + \lambda A(x) \bigr ) \Bigr ) \Bigl . \Bigr \rvert_{y = 0} 
		. 
	\end{align}
	Only terms with even powers of $\eps$ contribute, \ie $g_{n,k} = 0$ for all $n \in 2 \N_0 + 1$, $1 \leq k \leq n$. In particular we have $g_{0,0} = f$, $g_{1,0} = 0$, $g_{1,1} = 0$ and $f - g \in \Hoerrd{m - 3 \rho}$. 
\end{thm}
\begin{remark}
	The reason that only even powers of $\eps$ contribute can be traced back to the symmetry of $\GAe(x,y) = + \GAe(y,x)$. Note that this is consistent with what was said in the introduction, $\GAe([x,y]) = (y - x) \cdot \GAe([x,y])$ is indeed odd. 
\end{remark}
\begin{proof}
	The proof is virtually identical to the proof of Proposition~6.7 in \cite{IftimieMantiouPurice:magneticPseudodifferentialOperators:2005}; we will only specialize the formal part to the present case, the rigorous justification can be found in the reference. 
	
	For a symbol $f \in \Hoerrd{m}$, the integral kernel of its magnetic quantization is given by equation~\eqref{asympExp:magWQminSub:magneticIntKernelRewritten}. On the other hand, it is clear how to invert $\OpAel$ for $\lambda = 0$, $A \equiv 0$: we apply the non-magnetic Wigner transform $\mathcal{W}_{\eps} := \mathcal{W}^{A \equiv 0}_{\eps,\lambda=0}$ to the magnetic integral kernel of $\OpAel(f)$: 
	\begin{align*}
		\mathcal{W}_{\eps} K_{f,\eps,\lambda} (X) &= \int \dd y \, e^{-i y \cdot \xi} \, K_{f,\eps,\lambda} \bigl ( \tfrac{x}{\eps} + \tfrac{y}{2} , \tfrac{x}{\eps} - \tfrac{y}{2} \bigr ) \\ 
		&= \int \dd y \, \int \dd \eta \, e^{i y \cdot \eta} \, f \bigl ( x , \eta + \xi - \lambda \Gamma^A \bigl ( x + \tfrac{\eps}{2} \, y , x - \tfrac{\eps}{2} \, y \bigr ) \bigr ) 
	\end{align*}
	Since we have a separation of scales, we can expand $\Gamma^A \bigl (x + \tfrac{\eps}{2} \, y , x - \tfrac{\eps}{2} \, y \bigr )$ in powers of $\eps$ up to some \emph{even} $N$. We will find that only \emph{even} powers of $\eps$ survive -- which immediately explains the absence of the first-order correction, 
	\begin{align*}
		\Gamma^A \bigl ( x + \tfrac{\eps}{2} \, y , x - \tfrac{\eps}{2} \, y \bigr ) &= \int_{-\nicefrac{1}{2}}^{+\nicefrac{1}{2}} \dd s \, \Biggl ( \sum_{n = 0}^N \eps^n s^n \, \sum_{\sabs{a} = n} \partial_x^{a} A(x) \, y^{a} + R_N(s,x,y) \Biggr ) \\
		&= \sum_{n = 0}^{\nicefrac{N}{2}} \eps^{2n} \left ( \frac{1}{2} \right )^{2n} \frac{1}{2n + 1} \, \sum_{\sabs{a} = 2n} \partial_x^{a} A(x) \, y^{a} + \int_{-\nicefrac{1}{2}}^{+\nicefrac{1}{2}} \dd s \, R_N(s,x,y) 
		. 
	\end{align*}
	The remainder is bounded since it is the integral of a $\mathcal{C}^{\infty}_{\mathrm{pol}}$ function over the compact set $[-\nicefrac{1}{2} , + \nicefrac{1}{2}] \times [0,1]$. In any event, The exact value will not matter if we choose $N$ large enough as we set $y = 0$ in the end. 
	
	A Taylor expansion of $f \bigl ( x , \eta + \xi - \lambda \Gamma^A \bigl ( x + \tfrac{\eps}{2} \, y , x - \tfrac{\eps}{2} \, y \bigr ) \bigr )$ around $\eta - \lambda \Gamma^A$ and some elementary integral manipulations formally yield for the $n$th term of the expansion 
	\begin{align}
		\eps^n \, \sum_{k = 1}^n \lambda^k \, g_{n,k}(x,\xi - \lambda A(x)) 
		&= \sum_{\sabs{a} = n} 
		\frac{1}{a !} \, \bigl ( i \partial_y \bigr )^{a} \Bigl ( \partial_{\xi}^{a} f \bigl ( x , \xi - \lambda \Gamma^A(x + \tfrac{\eps}{2} \, y , x - \tfrac{\eps}{2} \, y) \bigr ) \Bigr ) \Bigl . \Bigr \rvert_{y = 0} 
		\label{asympExp:eqn:expansion_f_g_n_k}
	\end{align}
	where we substitute the expansion above for $\Gamma^A$. Each derivative in $y$ will give one factor of $\eps$, \ie we will have $n$ altogether. On the other hand, we have \emph{at least} $1$ and \emph{at most} $n$ factors of $\lambda$. Only even powers in $\eps$ contribute, because the expansion of $\Gamma^A \bigl ( x + \tfrac{\eps}{2} \, y , x - \tfrac{\eps}{2} \, y \bigr )$ contains only \emph{even} powers of $y$. Furthermore, all terms in this sum are bounded functions in $x$, because derivatives of $A$ are bounded by assumption. 
	\medskip

	\noindent
	To show that $g_{n,k}$ is in symbol class $\Hoerrd{m - (n+k) \rho}$, we need to have a closer look at equation~\eqref{asympExp:eqn:expansion_f_g_n_k}: the only possibility to get $k$ factors of $\lambda$ is to derive $\partial_{\xi}^{a} f \bigl ( x , \xi - \lambda \Gamma^A(x + \tfrac{\eps}{2} \, y , x - \tfrac{\eps}{2} \, y) \bigr )$ $k$ times with respect to $y$. Each of these $y$ derivatives becomes an additional derivative of $\partial_{\xi}^{a}f$ with respect to momentum. Hence, there is a total of $\abs{a} + k = n + k$ derivatives with respect to $\xi$. 
	\medskip

	\noindent
	The rigorous justification that these integrals exist can be found in \cite[Proposition~6.7]{IftimieMantiouPurice:magneticPseudodifferentialOperators:2005}. 
\end{proof}
\begin{remark}\label{asympExp:magWQminSub:epsExpansion}
	If we are interested in a one-parameter expansion in $\eps$ only, then 
	\begin{align*}
		g_n(X) := \eps^{-n} \sum_{\sabs{a} = n} \frac{1}{a !} \, \bigl ( i \partial_y \bigr )^{a} \Bigl ( \partial_{\xi}^{a} f \bigl ( x , \xi - \lambda \Gamma^A(x + \tfrac{\eps}{2} \, y , x - \tfrac{\eps}{2} \, y) + \lambda A(x) \bigr ) \Bigr ) \Bigl . \Bigr \rvert_{y = 0} 
	\end{align*}
	gives the $n$th order correction in $\eps$. 
\end{remark}
\begin{prop}[\cite{IftimieMantiouPurice:magneticPseudodifferentialOperators:2005}]\label{asympExp:thm:usualWQmagWQ}
	The converse statement also holds: if the magnetic field satisfies Assumption~\ref{asympExp:magWQminSub:assumption:stricterBA}, then for each $g \in \Hoerrd{m}$ there exists a unique $f \in \Hoerrd{m}$ such that $\Ope(g \circ \minsAl) = \OpAel(f)$, $f \asymp \sum_{n = 0}^{\infty} \sum_{k = 1}^n \eps^n \lambda^k \, f_{n,k}$, $f_{n,k} \in \Hoerrd{m - (n+k) \rho}$, can be expressed as a formal power series in $\eps$ where the $n$th term is given by 
	\begin{align}
		\sum_{k = 1}^n \lambda^k f_{n,k}(x,\xi) = \eps^{-n} \sum_{\abs{a} = n} \frac{1}{a !} \, (i \partial_{y})^{a} \bigl ( \partial_{\xi}^{a} f \bigr ) \bigl ( x , \xi + \lambda \GAe(x - \nicefrac{y}{2} , x + \nicefrac{y}{2}) - \lambda A(x) \bigr ) \bigl . \bigr \rvert_{y = 0}
	\end{align}
	In particular we have $f_{0,0} = g$, $f_{1,0} = 0$, $f_{1,1} = 0$ and $g - f \in \Hoerrd{m - 3 \rho}$. 
\end{prop}
\begin{proof}
	This proof works along the same lines: one magnetically Wigner-transforms the kernel of the operator $\Ope(f \circ \minsAl)$, we refer to  \cite[Proposition~6.9]{IftimieMantiouPurice:magneticPseudodifferentialOperators:2005} for details. 
\end{proof}
%


\section{Application to the Dirac equation} 
\label{Dirac}

To demonstrate the advantages of \emph{magnetic} Weyl calculus, we will apply it to a simple, yet interesting problem: the semirelativistic limit of the Dirac equation. This is a well-studied problem \cite{FoldyW:DiracTheory:1950,Thaller:DiracEquation:1992, Yndurain:relativisticQM:1996,Cordes:pseudodifferential_FW_transform:1983,Cordes:pseudodifferential_FW_transform:2004}, but we believe our derivation sheds a new light on origin of corrections. To keep this section readable and put emphasis on the computational aspects, we will dispense with mathematical rigor. Making these statements exact and putting them into context with previous works will be postponed to a future publication \cite{FuerstLein:nonRelLimitDirac:2008}. 

The dynamics of a relativistic spin-$\nicefrac{1}{2}$ particle with mass $m$ subjected to an electromagnetic field is described by the Dirac hamiltonian, 
\begin{align*}
	i \partial_t \Psi = \Bigl ( c^2 \, m \beta + c \, (-i \nabla_x) \cdot \alpha - e A(\hat{x}) \cdot \alpha + e V(\hat{x}) \Bigr ) \Psi 
	, 
	&& 
	\Psi \in L^2(\R^3,\C^4) 
	. 
\end{align*}
The hamiltonian consists of operator-valued matrices: $\alpha_j$, $j = 1,2,3$, has the $j$th Pauli matrix as entries in the offdiagonal, $\beta$ is the diagonal matrix with entries $1$, $1$, $-1$ and $-1$, namely,  
\begin{align*}
	\alpha_j = 
	\left (
	\begin{matrix}
		0 & \sigma_j \\
		\sigma_j & 0 \\
	\end{matrix}
	\right )
	, 
	&&
	\beta = 
	\left (
	\begin{matrix}
		\id_{\C^2} & 0 \\
		0 & - \id_{\C^2} \\
	\end{matrix}
	\right )
	. 
\end{align*}
As is customary, we have used shorthand notation for $\xi \cdot \alpha := \sum_{j = 1}^3 \xi_j \, \alpha_j$. If we assume that the components of the magnetic field $B = \dd A$ and an associated vector potential $A$ satisfy Assumption~\ref{asympExp:assumption:usualConditionsBA} and that $V \in \mathcal{BC}^{\infty}(\R^3)$, then Proposition~1.1 from \cite{Helffer_Nourrigat_Wang:spectre_lequation_Dirac:1989} as well as the fact that multiplication by $V$ defines a bounded operator on $L^2(\R^3,\C^4)$, the Dirac hamiltonian defines an essentially selfajoint operator on $\mathcal{C}_c^{\infty}(\R^3,\C^4)$. 

If we rescale the energy by $\nicefrac{1}{c^2}$ for convenience and absorb the charge into the definition of the potentials, we see that there are two \emph{natural} ways to write the Dirac hamiltonian, namely 
\begin{align*}
	\hat{H}_{\mathrm{D}} &= m \beta + \tfrac{1}{c} \bigl ( - i \nabla_x - \tfrac{1}{c} A(Q) \bigr ) \cdot \alpha + \tfrac{1}{c^2} V(Q) 
	\\
	&= m \beta + \bigl ( - i \tfrac{1}{c} \nabla_x - \tfrac{1}{c^2} A(Q) \bigr ) \cdot \alpha + \tfrac{1}{c^2} V(Q)
\end{align*}
where $Q := \hat{x}$ is the position operator. The first way of writing suggests to use 
\begin{align*}
	\PAc := - i \nabla_x - \tfrac{1}{c} A(Q) 
\end{align*}
as kinetic momentum operator, the second definition, 
\begin{align}
	\PiAc := - i \tfrac{1}{c} \nabla_x - \tfrac{1}{c^2} A(Q) = \tfrac{1}{c} \PAc 
	,
	\label{Dirac:eqn:kinetic_momentum}
\end{align}
absorbs an additional factor of $\nicefrac{1}{c}$. This seems nothing more than an algebraic trivialty, but is is this choice of building block operators which contains the physics. The first corresponds to the non-relativistic scaling where momenta are \emph{very} small and the $\nicefrac{1}{c} \rightarrow 0$ limit leads to the \emph{non}-relativistic limit. In \cite{FuerstLein:nonRelLimitDirac:2008} we derive the Pauli equation including fourth-order corrections with $\frac{1}{2m} \PiAc^2$ as kinetic energy operator in this scaling. This case is computationally more involved, than the second, the \emph{semirelativistic scaling}, which we will discuss now. The limit $\nicefrac{1}{c} \rightarrow 0$ will lead to the semirelativistic quantization where the kinetic energy operator is the magnetic quantization of $\sqrt{m^2 + \xi^2}$. The rescaled Dirac hamiltonian can be written as 
\begin{align}
	\hat{H}_{\mathrm{D}} = H_0(\PiAc) + \tfrac{1}{c^2} H_2(Q) 
	\label{Dirac:eqn:rescaled_Dirac_operator}
\end{align}
where 
\begin{align*}
	H_0(\xi) &:= m \beta + \xi \cdot \alpha \\
	H_2(x) &:= V(x) 
	. 
\end{align*}
Or to put another way, we can write $\hat{H}_{\mathrm{D}} = \OpAsr(\Hd)$ as the \emph{magnetic} quantization of the symbol $\Hd := H_0 + \tfrac{1}{c^2} H_2$ with respect to the pair of observables $(Q,\PiAc)$. The attentive reader will notice that we have defined the magnetic quantization of \emph{matrix-valued} symbols; we account for this by tensoring the Weyl system with the identity matrix $\id_{\C^4}$, 
\begin{align*}
	\OpAsr(\Hd) := \frac{1}{(2\pi)^3} \int \dd X \, \Fs^{-1}(\Hd) \, e^{i \sigma(X,(Q,\PiAc))} \otimes \id_{\C^4} 
	. 
\end{align*}
To put things into perspective, let us mention some works pertaining to the current approach: most of the previous treatments (\eg \cite{Hunziker:non_rel_limit_Dirac:1975,Grigore_Nenciu_Purice:non_rel_limit_Dirac:1989,Thaller:DiracEquation:1992,Yndurain:relativisticQM:1996}) are dealing with the derivation of the \emph{non}-relativistic limit where in our language the usual Pauli hamiltonian is the effective hamiltonian. Note that some of the authors use a different scaling to derive the non-relativistic limit: the prefactor of $A$ is taken to be $c$ and not $1$ in the original Dirac equation. In that case, the magnetic and electric field scale differently. 

The first to apply pseudodifferential techniques to the Dirac equation in order to obtain the \emph{semirelativistic} limit was Cordes \cite{Cordes:pseudodifferential_FW_transform:1983} who ordered the terms of the diagonalized hamiltonian by decay and not by powers of a small parameter. Physically, this is not satisfactory, because the prefactor decides which effects are and which are not measurable. Furthermore, it is not clear how to extend his ideas to allow for a non-relativistic limit. Brummelhuis and Nourrigat \cite{Brummelhuis_Nourrigat:scattering_Dirac:1999} also consider the problem with the help of pseudodifferential theory, they derive a power expansion in the semiclassical parameter $\hbar$ for approximate projections onto electronic and positronic state as well as the effective hamiltonian. They give the first-order correction to the effective hamiltonian explicitly and then continue with a semiclassical limit via an Egorov theorem. 

We adapt a technique developed by Panati, Spohn and Teufel, space-adiabatic perturbation theory \cite{PST:sapt:2002}, invented to cope with more general problems and due to its simplicity, the Dirac equation is one of the first systems this technique has been applied to \cite[Section~4.1]{Teufel:adiabaticPerturbationTheory:2003}. Their derivation rests on a very different physical mechanism, though: in their approach, the adiabatic decoupling is due to a slow variation of the electromagnetic potentials and in the limit $\eps \rightarrow 0$, the particle needs to travel farther and farther to see appreciable differences in the potentials. This limit is in fact equivalent to the case considered by Brummelhuis and Nourrigat \cite{Brummelhuis_Nourrigat:scattering_Dirac:1999} (see Appendix~\ref{appendix:equivalenceWeylSystems}). We, on the other hand, assume that the relativistic energy is small compared to the rest energy $m c^2$ of the particle and \emph{consequently}, the particle's velocity is small, $\nicefrac{v_0}{c} < 1$. While this sounds very similar to the point of view of Panati, Spohn and Teufel, the fields scale differently: the electromagnetic fields associated to the slowly-varying potentials $A(\eps x)$ and $V(\eps x)$ are of the order $\mathcal{O}(\eps)$ while the electromagnetic fields in our case are of order $\mathcal{O}(\nicefrac{1}{c^3})$ as we will see. Furthermore, we do not need to introduce another scale (the slow variation of the potentials) in addition to the energy scale given by $\nicefrac{1}{c}$. Our result holds almost in the entire range of validity of the Dirac equation: if the total energy of the particle approaches the pair creation threshold $2 m c^2$, we no longer expect the Dirac equation to give an accurate description of the physics anyway. 
\medskip

\noindent
Before we proceed, we give the first few terms of the asymptotic expansion in $\nicefrac{1}{c}$ of the magnetic Weyl product $\magBc$.

\subsection{Asymptotic expansion of $\magBc$} 
\label{Dirac:asymp_exp}

If we compare equation~\eqref{intro:observablesUsualScaling} with the definition of $\PiAc$, equation~\eqref{Dirac:eqn:kinetic_momentum}, we see that $\eps = \nicefrac{1}{c}$ and $\lambda = \nicefrac{1}{c^2}$. According to a simple modification of Theorem~\ref{appendix:equivalence_expansion}, we can write the expansion of $\magBc$ in terms of the two-parameter expansion of $\magWel$. The first few terms of $f \magBc g$ (with $f$ and $g$ being suitable matrix-valued functions, \eg matrix-valued Hörmander class symbols) are 
\begin{align}
	(f \magBc g)_{(0)} &= f \, g 
	\label{Dirac:eqn:magBc_expansion} 
	, 
	 \\
	(f \magBc g)_{(1)} &= - \tfrac{i}{2} \bigl \{ f , g \bigr \} 
	, 
	\notag \\
	(f \magBc g)_{(2)} &= - \tfrac{1}{4} \bigl ( \sigma(\nabla_Y,\nabla_Z) \bigr )^2 f(Y) \, g(Z) \bigl . \bigr \rvert_{Y = X = Z} 
	, 
	\notag \\
	(f \magBc g)_{(3)} &= \tfrac{i}{8} \bigl ( \sigma(\nabla_Y,\nabla_Z) \bigr )^3 f(Y) \, g(Z) \bigl . \bigr \rvert_{Y = X = Z} + \tfrac{i}{2} B_{lj}(x) \, \partial_{\xi_l} f(X) \, \partial_{\xi_j} g(X) 
	\notag 
	. 
\end{align}
While this seems very complicated, we will often need the product of two symbols which are functions of momentum only, $f \equiv f(\xi)$, $g \equiv g(\xi)$. In that case, only \emph{purely magnetic terms} (\ie $k_0 = 0$ in Theorem~\ref{intro:thm:main_result:asymptoticExpansion}) contribute,  
\begin{align}
	f \magBc g = f \, g + \tfrac{1}{c^3} \tfrac{i}{2} B_{lj} \partial_{\xi_l} f \, \partial_{\xi_j} g + \orderc{4} 
	= f \, g + \orderc{3}
	. 
	\label{Dirac:eqn:mag_product_two_momenta}
\end{align}
%


\subsection{Semirelativistic limit as adiabatic limit} 
\label{Dirac:semi_limit_adiabatic_limit}

The technique of choice, a modified version of \emph{space-adiabatic perturbation theory} \cite{PST:sapt:2002,PST:effDynamics:2003,Teufel:adiabaticPerturbationTheory:2003} that uses \emph{magnetic} Weyl calculus, rests on the interpretation of the semirelativistic limit $\nicefrac{1}{c} \rightarrow 0$ as an \emph{adiabatic limit}. This means, the Dirac hamiltonian has three characteristic features all adiabatic systems share, the so-called \emph{adiabatic trinity}: 
\begin{enumerate}[(i)]
	\item A distinction between \emph{slow and fast degrees of freedom}, \ie a decomposition of the original Hilbert space the hamiltonian acts on into $\mathcal{H} \cong \Hslow \otimes \Hfast$. Here, the fast Hilbert space is spanned by the electronic and the positronic state, $\Hfast \cong \C^2$. The slow Hilbert space is that of a \emph{non-}relativistic spin-$\nicefrac{1}{2}$ particle, $\Hslow \cong L^2(\R^3,\C^2)$. 
	\item A \emph{small}, dimensionless \emph{parameter} that quantifies the separation of scales. If $v_0$ is a typical velocity of the particle, we expect that no electron-positron pairs are created as long as $\nicefrac{v_0}{c} \ll 1$. However, for notational simplicity, we use $\nicefrac{1}{c}$ as small parameter. 
	\item A \emph{relevant part of the spectrum} of the \emph{unperturbed} operator, separated by a \emph{gap} from the remainder. If we consider the field-free case, then $H_0(- \nicefrac{i}{c} \, \nabla_x)$ fibers via the Fourier transform and the spectrum of each fiber hamiltonian is given by $\spec(H_0(\xi)) = \bigl \{ \pm \sqrt{m^2 + \xi^2} \bigr \}$. We are interested in the electronic subspace -- which is separated by a gap (of size $2 \sqrt{m^2 + \xi^2} \geq 2m$) from the positronic subspace. This ensures that even in the perturbed case, transitions from one band to the other are exponentially suppressed. 
\end{enumerate}
In a commutative diagram, the unperturbed situation looks as follows: 
\begin{align}
	\bfig
		\node L2C4(0,0)[L^2(\R^3,\C^4)]
		\node PiL2C4(0,-600)[\pi_0(- \nicefrac{i}{c} \, \nabla_x) \bigl ( L^2(\R^3,\C^4) \bigr )]
		\node L2C4FW(1200,0)[L^2(\R^3,\C^2) \otimes \C^2]
		\node L2C2FW(1200,-600)[L^2(\R^3,\C^2)]
		\arrow[L2C4`PiL2C4;\pi_0(- \nicefrac{i}{c} \, \nabla_x)]
		\arrow[L2C4`L2C4FW;u_0(- \nicefrac{i}{c} \, \nabla_x)]
		\arrow[L2C4FW`L2C2FW;\Piref]
		\arrow/-->/[PiL2C4`L2C2FW;]
		\Loop(0,0){L^2(\R^3,\C^4)}(ur,ul)_{e^{-i t H_0(- \nicefrac{i}{c} \, \nabla_x)}} 
		\Loop(1200,0){L^2(\R^3,\C^2) \otimes \C^2}(ur,ul)_{e^{-i t E(- \nicefrac{i}{c} \, \nabla_x) \beta}} 
		\Loop(1200,-600){L^2(\R^3,\C^2)}(dr,dl)^{e^{-i t E(- \nicefrac{i}{c} \, \nabla_x)}} 
	\efig
	\label{Dirac:diagram:unperturbed}
\end{align}
With a little abuse of notation, we will interpret all of these spaces as (subspaces of) $L^2(\R^3,\C^4) \cong L^2(\R^3) \otimes \C^4$ when convenient; operators acting on this space can be thought of as (operator-valued) $4 \times 4$ matrices or, if we are on the right-hand side of the diagram, as $2 \times 2$ matrices whose entries are itself (operator-valued) $2 \times 2$ matrices. The former identification is used during calculations, but the latter is conceptually useful. 

The objects in this diagram can be found in every text book on relativistic quantum mechanics (\eg \cite{Thaller:DiracEquation:1992,Yndurain:relativisticQM:1996}): $\pi_0$ is the projection onto the electronic subspace, 
\begin{align}
	\pi_0(\xi) &= \frac{1}{2} \left ( \id_{\C^4} + \frac{1}{E(\xi)} H_0(\xi) \right ) 
	. 
	\label{Dirac:eqn:pi0}
\end{align}
$u_0$ is the matrix-valued function that diagonalizes $H_0$, 
\begin{align}
	h_0 := u_0 \, H_0 \, {u_0}^{\ast} = \sqrt{m^2 + \xi^2} \beta =: E \beta 
	, 
	\label{Dirac:eqn:h0}
\end{align}
and `intertwines' $\pi_0$ with the \emph{reference projection}, 
\begin{align}
	u_0 \, \pi_0 \, {u_0}^{\ast} = \piref = \left (
	\begin{matrix}
		\id_{\C^2} & 0 \\
		0 & 0 \\
	\end{matrix}
	\right ) 
	, 
\end{align}
where 
\begin{align}
	u_0(\xi) &= \frac{1}{\sqrt{2E(E+m)}} \bigl ( (E+m) \id_{\C^4} - (\xi \cdot \alpha) \beta \bigr ) 
	. 
	\label{Dirac:eqn:u0}
\end{align}
The quantization of $\piref$ is $\Piref = \id_{L^2(\R^3,\C^2)} \otimes \piref$ projects out the positronic degrees of freedom in diagram~\eqref{Dirac:diagram:unperturbed}. If we are interested in the electron's dynamics only, we can describe it by an \emph{effective hamiltonian}, the quantization of 
\begin{align}
	h_{\mathrm{eff} \, 0} := \piref \, h_0 \, \piref = E \, \id_{\C^2} = \sqrt{m^2 + \xi^2} \, \id_{\C^2} 
	, 
\end{align}
in the following sense: 
\begin{align*}
	\Bigl ( e^{-i t H_0(- \nicefrac{i}{c} \, \nabla_x)} - {u_0}^{\ast}(- \nicefrac{i}{c} \, \nabla_x) \, e^{-i t E(- \nicefrac{i}{c} \, \nabla_x)} \, u_0(- \nicefrac{i}{c} \, \nabla_x) \Bigr ) \, \pi_0(- \nicefrac{i}{c} \, \nabla_x) = 0 
\end{align*}
Hence, we are able to relate the dynamics in the upper-left corner of diagram~\eqref{Dirac:diagram:unperturbed} with the reduced, effective dynamics in the lower-right corner. This reduction is possible as $H_0(- \nicefrac{i}{c} \, \nabla_x)$ and $\pi_0(- \nicefrac{i}{c} \, \nabla_x)$ commute, $\bigl [ H_0(- \nicefrac{i}{c} \, \nabla_x) , \pi_0(- \nicefrac{i}{c} \, \nabla_x) \bigr ] = 0$. Hence, the electronic subspace is invariant under the unperturbed dynamics. 
\medskip

\noindent
If we switch on the electromagnetic perturbation, this is no longer true, the commutator of $\hat{H}_{\mathrm{D}}$ and $\OpAsr(\pi_0) = \pi_0(\PiAc)$ is of order $\orderc{3}$. The immediate question is whether we can generalize diagram \eqref{Dirac:diagram:unperturbed} through some generalized projection $\Pi^c$ and generalized unitary $U^c$ such that 
\begin{align}
	\bfig
		\node L2C4(0,0)[L^2(\R^3,\C^4)]
		\node PiL2C4(0,-600)[\Pi^c \bigl ( L^2(\R^3,\C^4) \bigr )]
		\node L2C4FW(1200,0)[L^2(\R^3,\C^2) \otimes \C^2]
		\node L2C2FW(1200,-600)[L^2(\R^3,\C^2)]
		\arrow[L2C4`PiL2C4;\Pi^c]
		\arrow[L2C4`L2C4FW;U^c]
		\arrow[L2C4FW`L2C2FW;\Piref]
		\arrow/-->/[PiL2C4`L2C2FW;]
		\Loop(0,0){L^2(\R^3,\C^4)}(ur,ul)_{e^{-i t \hat{H}_{\mathrm{D}}}} 
		\Loop(1200,0){L^2(\R^3,\C^2) \otimes \C^2}(ur,ul)_{e^{-i t \OpAsr(h)}} 
		\Loop(1200,-600){L^2(\R^3,\C^2)}(dr,dl)^{e^{-i t \OpAsr(h_{\mathrm{eff}})}} 
	\efig
	\label{Dirac:diagram:perturbed}
\end{align}
holds. \emph{If} these objects exist, we require them to be an \emph{orthogonal projection} and a \emph{unitary} which commute with the full, perturbed Hamiltonian $\hat{H}_{\mathrm{D}}$ and block-diagonalize it, \ie 
\begin{align*}
	&{\Pi^c}^2 = \Pi^c, \; {\Pi^c}^{\ast} = \Pi^c 
	&& \bigl [ \OpAsr(\Hd),\Pi^c \bigr ] = 0 \\
	&{U^c}^{\ast} \, U^c = \id_{L^2(\R^3,\C^4)} , \; U^c \, {U^c}^{\ast} = \id_{L^2(\R^3,\C^2) \otimes \C^2} 
	&& U^c \, \Pi^c \, {U^c}^{\ast} = \Piref = \id_{L^2(\R^3)} \otimes \piref 
	. 
\end{align*}
Because of the last property $U^c$, is called \emph{intertwiner}. For suitable potentials $V$, we can \emph{translate} these equations (up to $\SemiHoermrd{-\infty}{1}{0}$) into equations of \emph{semiclassical symbols}. So \emph{if} there exist $\pi^c \in \SemiHoermrd{0}{1}{0}$ and $u^c \in \SemiHoermrd{0}{1}{0}$ such that $\Pi^c = \OpAsr(\pi^c) + \mathcal{O}_0(\nicefrac{1}{c^{\infty}})$\footnote{We say that two $c$-dependent bounded operators $A$ and $B$ on a Hilbert space $\mathcal{H}$ satisfy $A = B + \mathcal{O}_0(\nicefrac{1}{c^{\infty}})$ if for each $n \in \N_0$ there exists a constant $C_n$ such that $\norm{A - B}_{\mathcal{B}(\mathcal{H})} \leq C_n \, \tfrac{1}{c^n}$ } and $U^c = \OpAsr(u^c) + \mathcal{O}_0(\nicefrac{1}{c^{\infty}})$, then the corresponding symbols must satisfy 
\begin{align}
	&\pi^c \magBc \pi^c = \pi^c + \orderc{\infty}, \; 
	{\pi^c}^{\ast} = \pi^c 
	&&  \bigl [ \Hd , \pi^c \bigr ]_{\magBc} = \orderc{\infty} \label{Dirac:derivation:pilambdaConditions}\\
	&{u^c}^{\ast} \magBc u^c = \id_{\C^4} + \orderc{\infty} , \; 
	u^c \magBc {u^c}^{\ast} = \id_{\C^4} + \orderc{\infty} 
	&& u^c \magBc \pi^c \magBc {u^c}^{\ast} = \piref \label{Dirac:derivation:ulambdaConditions}
\end{align}
where the Moyal commutator is defined by $\bigl [ \Hd , \pi^c \bigr ]_{\magBc} := \Hd \magBc \pi^c - \pi^c \magBc \Hd$. 

If we incorporate magnetic Weyl calculus into space-adiabatic perturbation theory \cite{PST:sapt:2002,Teufel:adiabaticPerturbationTheory:2003}, we obtain an \emph{explicit resummation} for these symbols as well as formulas to correct $\OpAsr(\pi^c)$ and $\OpAsr(u^c)$ on the order $\mathcal{O}_0(\nicefrac{1}{c^{\infty}})$ to get a true projection and a true unitary in the operator sense. Then it is natural to assume that the principal symbols (the zeroth-order term) of the expansion of $\pi^c$ and $u^c$ have to be $\pi_0$ and $u_0$ -- the symbols of unitary and projection associated to the unperturbed hamiltonian. Starting from the unperturbed objects, Panati, Spohn and Teufel have found recursion relations which give corrections to $u_0$ and $\pi_0$ order-by-order in $\nicefrac{1}{c}$ which turn out to be \emph{independent} of the specific Weyl calculus used. 

The generator of the dynamics in the lower-right corner of diagram~\eqref{Dirac:diagram:perturbed} is the upper-left $2 \times 2$ submatrix of the diagonalized hamiltonian 
\begin{align}
	h := u^c \magBc \Hd \magBc {u^c}^{\ast}
	, 
\end{align}
\ie the \emph{effective hamiltonian} 
\begin{align}
	h_{\mathrm{eff}} := \piref \, h \, \piref 
	= \piref \, u^c \magBc \Hd \magBc {u^c}^{\ast} \, \piref 
	. 
\end{align}
There are technical and conceptual reasons for this specific choice that go beyond the scope of this text, we refer the interested reader to \cite[Section~3.3]{Teufel:adiabaticPerturbationTheory:2003} for details. The magnetic quantization of $h_{\mathrm{eff}}$ generates effective dynamics which approximate the full dynamics for electronic states, 
\begin{align*}
	\Bigl ( e^{- i t \OpAsr(\Hd)} - \OpAsr (u^c)^{\ast} \, e^{- i t \OpAsr(h_{\mathrm{eff}})} \, \OpAsr(u^c) \Bigr )  \, \OpAsr(\pi^c)  = \mathcal{O}_0 \bigl ( \abs{t} \nicefrac{1}{c^{\infty}} \bigr ) 
	. 
\end{align*}
%


\subsection{Effective hamiltonian} 
\label{Dirac:effective_hamiltonian}

In the present case, equation~\eqref{Dirac:eqn:mag_product_two_momenta} implies that the first correction to $\pi_0$ and $u_0$ is of \emph{third} order in $\nicefrac{1}{c}$: 
\begin{align*}
	&\pi_0 \magBc \pi_0 - \pi_0 = \orderc{3} 
	&& \bigl [ \Hd , \pi_0 \bigr ]_{\magBc} = \orderc{3} 
	\\
	&{u_0}^{\ast} \magBc u_0 = \id_{\C^4} + \orderc{3} , \; 
	u_0 \magBc {u_0}^{\ast} = \id_{\C^4} + \orderc{3} 
	&& u_0 \magBc \pi_0 \magBc {u_0}^{\ast} = \piref + \orderc{3} 
\end{align*}
From these equations, we conclude that $\pi_0$ and $u_0$ are an approximate Moyal projection and Moyal unitary, respectively, \ie $\pi^c = \pi_0 + \orderc{3}$ and $u^c = u_0 + \orderc{3}$. 

To compute the terms in the expansion of $h_{\mathrm{eff}}$ up to \emph{third} order in $\nicefrac{1}{c}$, we need to obtain the diagonalized hamiltonian symbol $h := u^c \magBc \Hd \magBc {u^c}^{\ast}$ up to \emph{second} order first. As expected, the leading-order term is the relativistic kinetic energy, 
\begin{align*}
	h_0 = \bigl ( u_0 \magBc H_0 \magBc {u_0}^{\ast} \bigr )_{(0)} = u_0 \, H_0 \, {u_0}^{\ast} = E \, \beta 
	. 
\end{align*}
If $h \asymp \sum_{n = 0}^{\infty} \tfrac{1}{c^n} h_n$ is the asymptotic expansion of the diagonalized hamiltonian, then we can determine $h_n$ recursively from $h \magBc u^c = u^c \magBc \Hd + \orderc{\infty}$: 
\begin{align}
	\tfrac{1}{c^n} (h_n \magBc u^c)_{(0)} &= \tfrac{1}{c^n} h_n \, u_0 + \orderc{n+1}
	\notag \\
	&= u^c \magBc H_{\mathrm{D}} - \bigl ( \mbox{$\sum_{k = 0}^{n-1}$} \, \tfrac{1}{c^k} \, h_k \bigr ) \magBc u^c 
		+ \orderc{n+1}
	\notag 
\end{align}
This simplifies calculations considerably. Starting from this equation, we arrive at the following formulas for $h_1$ and $h_2$: 
\begin{align*}
	h_1 &= \Bigl ( u_0 \, H_1 + u_1 \, H_0 - h_0 \, u_1 + (u_0 \magBc H_0)_{(1)} - (h_0 \magBc u_0)_{(1)} \Bigr ) \, {u_0}^{\ast} = 0 \\
	h_2 &= 
	\Bigl ( 
	u_0 \, H_2 + u_1 \, H_1 + u_2 \, H_0 - h_0 \, u_2 + 
	(u_0 \magBc H_1)_{(1)} + (u_1 \magBc H_0)_{(1)} - (h_0 \magBc u_1)_{(1)} - (h_1 \magBc u_0)_{(1)} + 
	\Bigr . \notag \\
	&\qquad \Bigl . + 
	(u_0 \magBc H_0)_{(2)} - (h_0 \magBc u_0)_{(2)} 
	\Bigr ) \, {u_0}^{\ast} 
	\\
	&= u_0 \, H_2 \, {u_0}^{\ast} = V \, \id_{\C^4}
\end{align*}
$h_1$ vanishes as expected and $h_2$ simplifies to $V$, because $u_1$, $u_2$ and $H_1$ vanish identically, and the product of two momentum-dependent functions contains no first- and second-order terms in $\nicefrac{1}{c}$ (equation~\eqref{Dirac:eqn:mag_product_two_momenta}). So far, we did not need to calculate one line explicitly to arrive at this result! The first three terms of the effective hamiltonian are obtained by sandwiching $h_0$ to $h_2$ with $\piref$. 
\begin{align*}
	h_{\mathrm{eff} \, 0} &= \piref \, h_0 \, \piref = E \, \id_{\C^2} = \sqrt{m^2 + \xi^2} \, \id_{\C^2} 
	\\
	h_{\mathrm{eff} \, 1} &= \piref \, h_1 \, \piref = 0
	\\
	h_{\mathrm{eff} \, 2} &= \piref \, h_2 \, \piref = V \, \id_{\C^2} 
\end{align*}
Finally, for $h_{\mathrm{eff} \, 3}$, we need to make some explicit computations and the first magnetic correction (third order in $\nicefrac{1}{c}$). There are three groups of terms which survive: 
\begin{align*}
	h_{\mathrm{eff} \, 3} &= 
	\piref \, h_3 \, \piref 
	= \piref \, \Bigl ( 
		u_3 \, H_0 - h_0 \, u_3 + 
		(u_0 \magBc H_2)_{(1)} - (h_2 \magBc u_0)_{(1)} + 
		(u_0 \magBc H_0)_{(3)} - (h_0 \magBc u_0)_{(3)}
	\Bigr ) \, {u_0}^{\ast} \, \piref \\
	&=: h_{\mathrm{eff} \, 30} + h_{\mathrm{eff} \, 31} + h_{\mathrm{eff} \, 33}
\end{align*}
The first two vanish when we project with $\piref$ from left and right, because $h_{\mathrm{eff} \, 0} = E \, \id_{\C^2}$ is a scalar symbol, 
\begin{align*}
	h_{\mathrm{eff} \, 30} 
	&= \piref \, \bigl ( u_3 \, H_0 - h_0 \, u_3 \bigr ) \, {u_0}^{\ast} \, \piref 
	 = \piref \, u_3 \, {u_0}^{\ast} \, u_0 \, H_0 \, {u_0}^{\ast} \, \piref - \piref \, h_0 \, u_3 \, {u_0}^{\ast} \, \piref 
	= 0
	. 
\end{align*}
The second and third group of terms need to be calculated explicitly; since the details are arithmetically intricate, we have moved them to Appendix~\ref{appendix:details_example}. $(u_0 \magBc H_2)_{(1)} - (h_2 \magBc u_0)_{(1)}$ gives a gradient coupling to the potential, 
\begin{align*}
	h_{\mathrm{eff} \, 31} &= \piref \, \bigl ( (u_0 \magBc H_2)_{(1)} - (h_2 \magBc u_0)_{(1)} \bigr ) \, \piref 
	= - i \piref \, \bigl \{ u_0 , V \bigr \} \, \piref 
	\\
	&
	= \frac{1}{2E (E+m)} (\nabla_x V \wedge \xi) \cdot \sigma 
	. 
\end{align*}
The last term, $h_{\mathrm{eff} \, 33}$, contains the spin-orbit coupling: 
\begin{align*}
	h_{\mathrm{eff} \, 33} &= \piref \, \bigl ( 
		(u_0 \magB H_0)_{(3)} - (h_0 \magB u_0)_{(3)}
	\bigr ) \, {u_0}^{\ast} \, \piref 
	= \tfrac{i}{2} B_{lj}(x) \, \piref \, \bigl ( 
		\partial_{\xi_l} u_0 \, \partial_{\xi_j} H_0 - \partial_{\xi_l} h_0 \, \partial_{\xi_j} u_0 
	\bigr ) \, {u_0}^{\ast} \, \piref
	\\
	&= - \frac{1}{2E} B \cdot \sigma 
\end{align*}
Altogether, the effective dynamics up to errors of fourth order in $\nicefrac{1}{c}$ are given by 
\begin{align}
	h_{\mathrm{eff}} = E \, \id_{\C^2} + \frac{1}{c^2} V \, \id_{\C^2} + \frac{1}{c^3} \biggl ( 
	\frac{1}{2E (E+m)} (\nabla_x V \wedge \xi) \cdot \sigma - \frac{1}{2E} B \cdot \sigma 
	\biggr ) + \orderc{4} 
	. 
\end{align}
The third-order correction is responsible for the spin dynamics and leads to the so-called T-BMT equation. This result has been previously derived by Cordes \cite{Cordes:pseudodifferential_FW_transform:1983} and Teufel \cite[Section~4.1]{Teufel:adiabaticPerturbationTheory:2003} under different hypothesis. We reiterate that the physical mechanism underlying the adiabatic decoupling in Teufel's work is different from the mechanism here. 

If we want to make this result rigorous, we will have to explicitly show that the construction of space-adiabatic perturbation theory still works when one replaces usual Weyl calculus with magnetic Weyl calculus. This has been the motivation for making the two-parameter expansion rigorous in the first place, but deserves a publication in its own right \cite{FuerstLein:nonRelLimitDirac:2008}. 


%
\begin{appendix}
	\section{Equivalence of Weyl systems in both scalings} 
	\label{appendix:equivalenceWeylSystems}

	\begin{lem}
		The adiabatic scaling and the usual scaling are related by the unitary $U_{\eps}$, $\bigl ( U_{\eps} \varphi \bigr )(x) := \eps^{- \nicefrac{d}{2}} \varphi \bigl ( \tfrac{x}{\eps} \bigr )$, $\varphi \in L^2(\R^d)$, \ie we have 
		\begin{align*}
			Q &= U_{\eps} \, \Qe \, U_{\eps}^{-1} \\
			P^A_{\eps,\lambda} &= U_{\eps} \, \Pi^A_{\eps,\lambda} \, U_{\eps}^{-1} 
			. 
		\end{align*}
	\end{lem}
	\begin{proof}
		Let $\varphi \in L^2(\R^d)$. Then we have for $\Qe$ 
		\begin{align*}
			(U_{\eps} \, \Qe \, U_{\eps}^{-1} U_{\eps} \varphi )(x) &= \bigl ( U_{\eps} \Qe \varphi \bigr ) (x) = \eps^{-\nicefrac{d}{2}} \, \bigl ( \Qe \varphi \bigr ) \bigl ( \tfrac{x}{\eps} \bigr ) 
			\\
			&= \eps^{-\nicefrac{d}{2}} \, \eps \tfrac{x}{\eps} \varphi \bigl ( \tfrac{x}{\eps} \bigr ) 
			= Q \, \bigl ( U_{\eps} \varphi \bigr ) (x) . 
		\end{align*}
		Similarly, we get for the momentum operators 
		\begin{align*}
			(U_{\eps} \, \Pi^A_{\eps,\lambda} \, U_{\eps}^{-1} U_{\eps} \varphi )(x) 
			&= \bigl ( U_{\eps} \Pi^A_{\eps,\lambda} \varphi \bigr ) (x) 
			= \eps^{-\nicefrac{d}{2}} \, \bigl ( \Pi^A_{\eps,\lambda} \varphi \bigr ) \bigl ( \tfrac{x}{\eps} \bigr ) 
			\\
			&= \eps^{-\nicefrac{d}{2}} \, \bigl ( -i (\nabla_x \varphi) \bigl ( \tfrac{x}{\eps} \bigr ) - \lambda A \bigl (\eps \tfrac{x}{\eps} \bigr ) \, \varphi \bigl ( \tfrac{x}{\eps} \bigr ) \bigr ) 
			= \bigl ( - i \eps \nabla_x - \lambda A(Q) \bigr ) \, \bigl ( U_{\eps} \varphi \bigr ) (x) 
			. 
		\end{align*}
		Hence the two scalings are unitarily equivalent. 
	\end{proof}
	\begin{cor}
		The Weyl systems associated to the two scalings given by equations \eqref{intro:observablesUsualScaling} and \eqref{intro:observablesAdiabaticScaling} are unitarily equivalent. 
	\end{cor}
	\begin{thm}\label{appendix:equivalence_expansion}
		The asymptotic two-parameter expansions of the magnetic Weyl products with respect to either scaling are given by the same terms order-by-order in $\eps$ and $\lambda$. 
	\end{thm}
	\begin{proof}
		To show that the asymptotic expansion of the product is the same, we have to revisit Theorem~\ref{asympExp:thm:equivalenceProduct} (proof of equivalence of the two non-asymptotic product formulas product formulas) and translate the relevant formulas to the usual scaling. It suffices to show that the twister in both cases is the same function and thus the expansion has to be identical, too. We denote magnetic Weyl quantization with respect to the Weyl system in usual scaling, $\WeylUs(X) := e^{i \sigma(X,(Q,P^A_{\eps,\lambda}))}$, with $\OpAus$. For convenience of the reader, we will follow the notation in the proof of Theorem \ref{intro:thm:main_result:asymptoticExpansion} as closely as possible. 
		
		With a simple scaling argument, we get the composition rule for the Weyl system $\WeylUs(X)$: 
		\begin{align*}
			\WeylUs(Y) \, \WeylUs(Z) 
			&= W^{\nicefrac{\lambda}{\eps} \, A}(\eps y , \eta) \, W^{\nicefrac{\lambda}{\eps} \, A}(\eps z , \zeta) 
			\\
			&= e^{\frac{i}{2} \sigma((\eps y , \eta),(\eps z , \zeta))} \, \Omega^{\nicefrac{\lambda}{\eps} \, B}(Q,Q+\eps y , Q + \eps y + \eps z) \, W^{\nicefrac{\lambda}{\eps} \, A}(\eps y + \eps z , \eta + \zeta) \\
			&= e^{i \frac{\eps}{2} \sigma(Y,Z)} \OBel(Q,Q+\eps y,Q+\eps y + \eps z) \WeylUs(Y+Z)
		\end{align*}
		In Step 1 of the proof, we conclude from the composition law of the Weyl system (reformulated in the usual scaling), 
		\begin{align*}
			\OpAus(f) \, \OpAus(g) = \frac{1}{(2\pi)^{2d}} \int \dd Z \biggl ( \int \dd Y \, &\bigl ( \Fs^{-1} f \bigr )(Y) \, \bigl ( \Fs^{-1} g \bigr )(Z-Y) \;  e^{i \tfrac{\eps}{2} \, \sigma(Y,Z)} \biggr . \cdot \\
			&\biggl . \qquad \cdot \OBel(Q,Q+\eps y,Q+\eps z) \biggr ) \WeylUs(Z) 
			, 
		\end{align*}
		that we need to find the operator kernel for 
		\begin{align*}
			\OBel(Q,Q+\eps y,Q+\eps z) \WeylUs(Z) 
			. 
		\end{align*}
		If we apply this operator to a function $\varphi \in L^2(\R^d)$, we obtain 
		\begin{align*}
			\bigl ( \OBel&(Q,Q+\eps y,Q+\eps z) \WeylUs(Z) \varphi \bigr )(v) = \OBel(v,v+\eps y,v + \eps z) \, e^{-i (v + \nicefrac{\eps}{2} \, z) \cdot \zeta} \, e^{- i \nicefrac{\lambda}{\eps} \Gamma^A([v,v+\eps z])} \, \varphi(v + \eps z) \\
			&= \int \dd u \, e^{-i (u - \nicefrac{\eps}{2} \, z) \cdot \zeta} \, e^{-i \nicefrac{\lambda}{\eps} \Gamma^A([u-\eps z,u])} \, \OBel(u - \eps z , u+\eps y , u) \, \delta \bigl ( u - (v + \eps z) \bigr ) \, \varphi(u) \\
			&=: \int \dd u \, \tilde{K}(y,Z;u,v) \, \varphi(u)
			. 
		\end{align*}
		To find the symbol associated to this object, we employ the Wigner transform adapted to observables in the usual scaling defined by 
		\begin{align*}
			\WignerUs (\varphi,\psi)(X) :=& \eps^d \, \bigl ( \Fs \expval{\varphi , \WeylUs(\cdot) \psi} \bigr )(-X) 
			\\
			=& \int \dd y \, e^{- i y \cdot \xi} e^{-i \nicefrac{\lambda}{\eps} \, \GA([x - \nicefrac{\eps}{2} \, y,x + \nicefrac{\eps}{2} \, y ])} \, \varphi^{\ast} \bigl ( x - \tfrac{\eps}{2}y \bigr ) \, \psi \bigl ( x + \tfrac{\eps}{2} y \bigr ) 
			. 
		\end{align*}
		If we apply this to the integral kernel above, we get by the essentially the same calculation as before, 
		\begin{align*}
			\bigl ( \WignerUs &\tilde{K}(y,Z;\cdot,\cdot) \bigr )(X) = \eps^d \, \int \dd u \, e^{- i u \cdot \xi} e^{-i \nicefrac{\lambda}{\eps} \, \GA([x - \nicefrac{\eps}{2} \, u,x + \nicefrac{\eps}{2} \, u ])} \, \tilde{K} \bigl ( y,Z;x - \tfrac{\eps}{2} u , x + \tfrac{\eps}{2} u \bigr ) \\
			&= \eps^d \, \eps^{-d} \int \dd u \, e^{- i u \cdot \xi} e^{-i \nicefrac{\lambda}{\eps} \, \GA([x - \nicefrac{\eps}{2} \, u,x + \nicefrac{\eps}{2} \, u ])} \, 
			e^{-i (x - \nicefrac{\eps}{2} \, u - \nicefrac{\eps}{2} \, z) \cdot \eta} \, e^{-i \nicefrac{\lambda}{\eps} \Gamma^A([x - \nicefrac{\eps}{2} \, u-\eps z,x - \nicefrac{\eps}{2} \, u])} \cdot \\
			&\qquad \qquad \qquad \cdot \OBel \bigl ( x - \tfrac{\eps}{2} \, u - \eps z , x - \tfrac{\eps}{2} \, u+\eps y , x - \tfrac{\eps}{2} \, u \bigr ) \, \delta (z+u) 
			 \\ 
			&= e^{i \sigma(X,Z)} \OBel \bigl (x - \tfrac{\eps}{2} \, z , x + \eps \bigl ( y - \tfrac{z}{2} \bigr ) , x + \tfrac{\eps}{2} \, z \bigr )
			. 
		\end{align*}
		If we plug this into the remainder of the proof, we see that the twister term (after replacing $z$ with $y+z$ just as in Step 3) obtained here is identical to the one obtained in the adiabatic scaling, 
		\begin{align*}
			e^{i \tfrac{\eps}{2} \sigma(X,Y)} \, \OBel \bigl (x - \tfrac{\eps}{2} \, (y+z) , x + \tfrac{\eps}{2} (y-z) , x + \tfrac{\eps}{2} \, (y+z) \bigr) 
			. 
		\end{align*}
		Hence the two expansions need to agree. 
	\end{proof}
	%
	

	\section{Expansion of the twister} 
	\label{app:formalExpansionTwister}

	\begin{lem}\label{app:lemmaFormalExpansionTwister}
		Assume $B$ satisfies Assumption~\ref{asympExp:assumption:usualConditionsBA}. Then we can expand $\gBe$ around $x$ to arbitrary order $N$ in powers of $\eps$: 
		\begin{align}
			\gBe(x,y,z) &= - \sum_{n = 1}^{N} \frac{\eps^n}{n!} \, \partial_{x_{j_1}} \cdots \partial_{x_{j_{n-1}}} B_{kl}(x) \, y_k \, z_l \, \left ( - \frac{1}{2} \right )^{n+1} \frac{1}{(n+1)^2} \sum_{c = 1}^n \noverk{n+1}{c} \, \cdot \notag \\
			&\qquad \qquad \cdot \bigl ( (1 - (-1)^{n+1}) c - (1 - (-1)^{c}) (n+1) \bigr ) \, y_{j_1} \cdots y_{j_{c-1}} z_{j_{c}} \cdots z_{j_{n-1}} \notag + R_N[\gBe](x,y,z) \\
			&=: - \sum_{n = 1}^{N} \eps^n \sum_{\abs{\alpha} + \abs{\beta} = n-1} C_{n,\alpha,\beta} \, \partial_x^{\alpha} \partial_x^{\beta} B_{kl}(x) \, y_k z_l \, y^{\alpha} \, z^{\beta} + R_N[\gBe](x,y,z) \\
			&=: - \sum_{n = 1}^{N} \eps^n \, {\mathcal{L}}_n + R_N[\gBe](x,y,z) 
		\end{align}
		In particular, the flux is of order $\eps$ and the $n$th-order term is a sum of monomials in position of degree $n+1$ and each of the terms is a $\mathcal{BC}^{\infty}(\R^d_x,\mathcal{C}^{\infty}_{\mathrm{pol}}(\R^d_y \times \R^d_z))$ function. The remainder is a $\mathcal{BC}^{\infty}(\R^d,\mathcal{C}^{\infty}_{\mathrm{pol}}(\R^d \times \R^d))$ function that is $\ordere{N+1}$ and can be explicitly written as a bounded function of $x$, $y$ and $z$ as well as $N+2$ factors of $y$ and $z$. 
	\end{lem}
	\begin{proof}
		We choose the transversal gauge to represent $B$, \ie 
		\begin{align}
			A_l(x + a) &= - \int_0^1 \dd s \, B_{lj} (x + s a) \, s a_j 
			\label{appendix:eqn:transversal_gauge}
		\end{align}
		and rewrite the flux integral into three line integrals over the edges of the triangle. 
		\begin{align*}
			\gBe(x,y,z) &= \frac{1}{\eps} \int_0^1 \dd t \Bigl [ \eps \, (y_l + z_l) \, A_l \bigl (x + \eps (t - \nicefrac{1}{2}) (y+z) \bigr ) + \Bigr . \\
			&\qquad \qquad \Bigl . 
			- \eps \, y_l \, A_l \bigl ( x + \eps (t - \nicefrac{1}{2}) y - \tfrac{\eps}{2} z \bigr ) 
			- \eps \, z_l \, A_l \bigl ( x + \tfrac{\eps}{2} y + \eps (t - \nicefrac{1}{2}) z \bigr ) \Bigr ] 
			\\
			&= \eps \int_{-\nicefrac{1}{2}}^{+\nicefrac{1}{2}} \dd t \int_0^1 \dd s \, s \Bigl [ - B_{lj} \bigl ( x + \eps s t (y + z) \bigr ) \, (y_l + z_l) \, t (y_j + z_j) +  \Bigr. \\
			&\qquad \qquad \Bigl .  
			+ B_{lj} \bigl ( x + \eps s \bigl ( t y - \tfrac{z}{2} \bigr ) \bigr ) \, y_l \, \bigl ( t y_j - \tfrac{z_j}{2} \bigr ) 
			+ B_{lj} \bigl ( x + \eps s \bigl ( \tfrac{y}{2} + t z \bigr ) \bigr ) \, z_l \, \bigl ( \tfrac{y_j}{2} + t z_j \bigr ) \Bigr ]
		\end{align*}
		All these terms have a prefactor of $\eps$ which stems from the explicit expression of transversal gauge. We will now Taylor expand each of the three terms up to $N-1$th order around $x$ (so that it is of $N$th order in $\eps$). 
		\begin{align*}
			\int_{-\nicefrac{1}{2}}^{+\nicefrac{1}{2}} &\dd t \, \int_0^1 \dd s \, B_{lj} \bigl ( x + \eps s t (y + z) \bigr ) \, \eps s (y_{j} + z_{j} ) = \\
			&= \int_0^1 \dd s \, \int_{-\nicefrac{1}{2}}^{+\nicefrac{1}{2}} \dd t \, \sum_{n = 0}^{N-1} \frac{\eps^{n+1}}{n!} s^{n+1} s^{-n} \, \partial_{x_{j_1}} \cdots \partial_{x_{j_n}} B_{lj_{n+1}}(x) \, t^{n+1} \, \prod_{m = 1}^{n+1} (y_{j_m} + z_{j_m}) + R_{1 \, N \, l}(x,y,z)
		\end{align*}
		The remainder $R_{1 \, N \, l}$ is of order $N+1$ in $\eps$, bounded in $x$ and polynomially bounded in $y$ and $z$. It is a sum of monomials in $y$ and $z$ of degree $N+1$. 
		\begin{align*}
			R_{1 \, N \, l}&(x,y,z) = \int_0^1 \dd \tau \int_0^1 \dd s  \int_{-\nicefrac{1}{2}}^{+\nicefrac{1}{2}} \dd t \, \frac{1}{(N-1)!} (1-\tau)^{N-1} \, \partial_{\tau}^{N} B_{lj}(x + \eps \tau s t (y+z) ) \, \eps s t (y + z) 
			\\
			&= \eps^{N+1} \, \sum_{\sabs{\alpha} = N} \frac{N}{\alpha !} \int_0^1 \dd s  \int_{-\nicefrac{1}{2}}^{+\nicefrac{1}{2}} \dd t \, s \, t^{N+1} \, (y + z)^{\alpha} \, (y_j + z_j) \, 
			\int_0^1 \dd \tau \, (1 - \tau)^{N-1} \, \partial_x^{\alpha} B_{lj} (x + \eps \tau st (y+z)) 
		\end{align*}
		The $n$th order term in $\eps$ (the $n-1$th term of the Taylor expansion) reads 
		\begin{align*}
			\frac{\eps^{n}}{(n-1)!} &\int_0^1 \dd s \, s \, \int_{-\nicefrac{1}{2}}^{+\nicefrac{1}{2}} \dd t \, t^{n} \, \partial_{x_{j_1}} \cdots \partial_{x_{j_{n-1}}} B_{lj_{n}}(x) \, \prod_{m = 1}^{n} (y_{j_m} + z_{j_m}) = 
			\\
			&= \frac{1}{2} \frac{\eps^{n}}{n!} \, \left ( \frac{1}{2} \right )^{n+1} \frac{1 + (-1)^{n}}{n+1} \, \partial_{x_{j_1}} \cdots \partial_{x_{j_n}} B_{lj_{n}}(x) \, \sum_{m = 0}^{n} \noverk{n}{m} y_{j_1} \cdots y_{j_m} z_{j_{m+1}} \cdots z_{j_{n}} 
			. 
		\end{align*}
		The other factors can be calculated in the same fashion: 
		\begin{align*}
			\frac{\eps^{n}}{(n-1)!} &\int_{-\nicefrac{1}{2}}^{+\nicefrac{1}{2}} \dd t \, \int_0^1 \dd s \, s^{n} s^{-(n-1)} \, \partial_{x_{j_1}} \cdots \partial_{x_{j_{n-1}}} B_{lj_{n}}(x) \, \prod_{m = 1}^{n} \bigl ( t y_{j_m} - \tfrac{1}{2} z_{j_m} \bigr ) = 
			\\
			&= \frac{\eps^{n}}{n!} \left ( \frac{1}{2} \right )^{n+2} \, \partial_{x_{j_1}} \cdots \partial_{x_{j_{n-1}}} B_{lj_{n}}(x) \, \sum_{m = 0}^{n} \noverk{n}{m} \, \frac{(-1)^{n-m} +(-1)^{n}}{m+1} \, y_{j_1} \cdots y_{j_m} z_{j_{m+1}} \cdots z_{j_{n}} 
		\end{align*}
		The remainder is also of the correct order in $\eps$, contains $N+2$ $q$s and a $\mathcal{BC}^{\infty}(\R^d,\mathcal{C}^{\infty}_{\mathrm{pol}}(\R^d \times \R^d))$ function as prefactor: 
		\begin{align*}
			R_{2 \, N \, l}(x,y,z) &= \eps^{N+1} \, \sum_{\sabs{\alpha} = N} \frac{N}{\alpha !} \int_0^1 \dd s  \int_{-\nicefrac{1}{2}}^{+\nicefrac{1}{2}} \dd t \, s \, \bigl ( t y - \tfrac{z}{2} \bigr )^{\alpha} \, \bigl ( t y_j - \tfrac{z_j}{2} \bigr ) 
			\cdot \\
			& \qquad \cdot 
			\int_0^1 \dd \tau \, (1 - \tau)^{N-1} \, \partial_x^{\alpha} B_{lj} \bigl ( x + \eps \tau \bigl ( st y - \tfrac{z}{2} \bigr ) \bigr ) 
		\end{align*}
		The last term satisfies the same properties as $R_{1 \, N \, l}$: 
		\begin{align*}
			\int_{-\nicefrac{1}{2}}^{+\nicefrac{1}{2}} \dd t &\int_0^1 \dd s \, \frac{\eps^{n}}{(n-1)!} s^{n} s^{-(n-1)} \, \partial_{x_{j_1}} \cdots \partial_{x_{j_{n-1}}} B_{lj_{n}}(x) \, \prod_{m = 1}^{n} \bigl ( \tfrac{1}{2} y_{j_m} + t z_{j_m} \bigr ) = 
			\\
			&= \frac{\eps^{n}}{n!}  \, \left ( \frac{1}{2} \right )^{n+2} \, \partial_{x_{j_1}} \cdots \partial_{x_{j_{n-1}}} B_{lj_{n}}(x) \, \sum_{m = 0}^{n} \noverk{n}{m} \, \frac{1 + (-1)^{n-m}}{n+1-m} \, y_{j_1} \cdots y_{j_m} z_{j_{m+1}} \cdots z_{j_{n}} 
		\end{align*}
		$R_{3 \, N \, l}$ satisfies the same properties as $R_{1 \, N \, l}$ and $R_{2 \, N \, l}$, 
		\begin{align*}
			R_{3 \, N \, l}(x,y,z) &= \eps^{N+1} \, \sum_{\sabs{\alpha} = N} \frac{N}{\alpha !} \int_0^1 \dd s  \int_{-\nicefrac{1}{2}}^{+\nicefrac{1}{2}} \dd t \, s \, \bigl ( \tfrac{y}{2} + t z \bigr )^{\alpha} \, \bigl ( \tfrac{y_j}{2} + t z_j \bigr )  
			\cdot \\
			& \qquad \cdot 
			\int_0^1 \dd \tau \, (1 - \tau)^{N-1} \, \partial_x^{\alpha} B_{lj} \bigl ( x + \eps \tau s \bigl ( \tfrac{y}{2} + t z \bigr ) \bigr ) 
			. 
		\end{align*}
		Put together, we obtain for the $n$th order term: 
		\begin{align*}
			\frac{1}{2} &\frac{\eps^{n}}{n!} \left ( \frac{1}{2} \right )^{n+1} \, \partial_{x_{j_1}} \cdots \partial_{x_{j_{n-1}}} B_{lj_{n}}(x) \, \sum_{m = 0}^{n} \noverk{n}{m} \cdot \\
			&\qquad \cdot \left [ \frac{1 + (-1)^{n}}{n+1} (y_l + z_l) - \frac{(-1)^{n-m} + (-1)^{n}}{m+1} y_l - \frac{1 + (-1)^{n-m}}{n+1-m} z_l \right ] \, y_{j_1} \cdots y_{j_m} z_{j_{m+1}} \cdots z_{j_{n}} \\
			&= \frac{\eps^{n}}{n!} \left ( - \frac{1}{2} \right )^{n+1} \frac{1}{(n+1)^2} \, \sum_{\sabs{\alpha} + \sabs{\beta} = n-1} \partial_x^{\alpha} \partial_x^{\beta} B_{lk}(x) \, y_l z_k \cdot \\
			&\qquad \cdot \noverk{n+1}{\sabs{\alpha}+1} \bigl ( ( 1 - (-1)^{\sabs{\alpha}+1} ) (n+1) - (1 - (-1)^{n+1}) (\sabs{\alpha} + 1) \bigr ) \, y^{\alpha} z^{\beta}
		\end{align*}
		The total remainder of the expansion reads 
		\begin{align*}
			R_N[\gBe] &= y_l \, \bigl ( R_{1 \, N \, l} - R_{2 \, N \, l} \bigr ) + z_l \, \bigl ( R_{1 \, N \, l} - R_{3 \, N \, l} \bigr ) \in \mathcal{BC}^{\infty}(\R^d,\mathcal{C}^{\infty}_{\mathrm{pol}}(\R^d \times \R^d))
			. 
		\end{align*}
		In total, the remainder is a sum of monomials with bounded coefficients of degree $N+2$ while it is of $\ordere{N+1}$. 
	\end{proof}
	%
	

	\section{Properties of derivatives of $\gBe$} 
	\label{app:propertiesDerivativesgBe}
	
	For convenience, we give two theorems found in \cite{IftimieMantiouPurice:magneticPseudodifferentialOperators:2005} on the magnetic flux and its expontential which are needed to make the expansion rigorous: 
	\begin{lem}\label{appendix:boundednessLemmagBel}
		If the magnetic field $B_{lj}$, $1 \leq l,j \leq n$, satisfies the usual conditions, then 
		\begin{align*}
			\partial_{x_j} \gBe &= D_{jk}(x,y,z) \, y_k + E_{jk}(x,y,z) \, z_k \\
			\partial_{y_j} \gBe &= D'_{jk}(x,y,z) \, y_k + E'_{jk}(x,y,z) \, z_k \\
			\partial_{z_j} \gBe &= D''_{jk}(x,y,z) \, y_k + E''_{jk}(x,y,z) \, z_k 
		\end{align*}
		where the coefficients $D_{jk} , \ldots , E''_{jk} \in \mathcal{BC}^{\infty}(\R^d \times \R^d \times \R^d)$, $1 \leq j,k \leq d$. 
	\end{lem}
	\begin{proof}
		The corners of the flux triangles of $F_B$ found in \cite{IftimieMantiouPurice:magneticPseudodifferentialOperators:2005} differ from those of $\gBe$, but the proof carries over with trivial modifications. 
	\end{proof}
	A direct consequence of this is the following simple corollary: 
	\begin{cor}\label{appendix:boundednessCoroBel}
		If the magnetic field satisfies the usual conditions, then 
		\begin{align*}
			\babs{\partial_x^a \partial_y^b \partial_z^c e^{- i \lambda \gBe(x,y,z)}} \leq C_{abc} \bigl ( \sexpval{y} + \sexpval{z} \bigr )^{\sabs{a} + \sabs{b} + \sabs{c}} 
			\leq \tilde{C}_{abc} \sexpval{y}^{\sabs{a} + \sabs{b} + \sabs{c}} \sexpval{z}^{\sabs{a} + \sabs{b} + \sabs{c}} 
			&& \forall a,b,c \in {\N_0}^d
			, 
		\end{align*}
		\ie derivatives of $e^{- i \lambda \gBe(x,y,z)}$ are $\mathcal{C}^{\infty}_{\mathrm{pol}}$ functions in $y$ and $z$. 
	\end{cor}
	%
	

	\section{Existence of oscillatory integrals} 
	\label{appendix:existenceOscInt}
	
	To derive the adiabatic expansion, we have to ensure the existence of two types of oscillatory integrals, one is relevant for the $(n,k)$ term of the two-parameter expansion, the other is necessary to show existence of remainders and the $k$th term of the $\lambda$ expansion. 
	\begin{lem}\label{appendix:existenceOscInt:Lemma1}
		Let $f \in \Hoerrd{m}$, $\rho \in [0,1]$. Then for all multiindices $a, \alpha \in \N_0^d$ 
		\begin{align}
			G(X) := \frac{1}{(2\pi)^d} \int \dd Y \, e^{i \sigma(X,Y)} y^a \eta^{\alpha} (\Fs^{-1} f)(Y) = \bigl ( (-i \partial_{\xi})^{a} (+ i \partial_x)^{\alpha} f \bigr )(X) 
		\end{align}
		exists as an oscillatory integral and is in symbol class $\Hoerrd{m - \sabs{a} \rho}$. 
	\end{lem}
	\begin{proof}
		Since $f$ is a function of tempered growth, we can consider it as an element of $\Schwartz'(\R^{2d})$. Then, we can rewrite $G$ as $G = \Fs \hat{x}^a \hat{\xi}^{\alpha} \Fs$ where $\hat{x}$ and $\hat{\xi}$ are the multiplication operators initially defined on $\Schwartz(\R^{2d})$ which are extended to tempered distributions by duality. Then for any $\varphi \in \Schwartz(\R^{2d})$, we have 
		\begin{align*}
			\bigl ( G , \varphi \bigr ) &= \bigl ( \Fs \hat{x}^a \hat{\xi}^{\alpha} \Fs f , \varphi \bigr ) 
			= \bigl ( f , \Fs \hat{x}^a \hat{\xi}^{\alpha} \Fs \varphi \bigr ) 
			= \bigl ( f , (+i \partial_{\xi})^a (-i \partial_x)^{\alpha} \varphi \bigr ) 
			\\
			&= \bigl ( (-i \partial_{\xi})^a (+i \partial_x)^{\alpha} f , \varphi \bigr ) 
		\end{align*}
		where $( \cdot , \cdot )$ denotes the usual duality bracket. 
		
		Thus, the integral exists as an oscillatory integral. $G = (-i \partial_{\xi})^a (+i \partial_x)^{\alpha} f$ is also in the correct symbol class, namely $\Hoerrd{m - \sabs{a} \rho}$, and the lemma has been proven. 
	\end{proof}
	The next corollary is an immediate consequence and contains the relevant result for the term-by-term expansion of the magnetic product. 
	\begin{cor}\label{appendix:existenceOscInt:Lemma2}
		Let $f \in \Hoerrd{m_1}$, $g \in \Hoerrd{m_2}$, $\rho \in [0,1]$ and $a , \alpha, b, \beta \in \N_0^d$ be arbitrary multiindices. Then for all functions $B \in \mathcal{BC}^{\infty}(\R^d)$ the oscillatory integral 
		\begin{align}
			G(X) := \frac{1}{(2\pi)^{2d}} \int \dd Y \int \dd Z e^{i \sigma(X,Y+Z)} B(x) \, y^a \eta^{\alpha} (\Fs^{-1} f)(Y) \, z^b \zeta^{\beta} (\Fs^{-1} g)(Z) 
		\end{align}
		exists, is in symbol class $\Hoerrd{m_1 + m_2 - (\sabs{a} + \sabs{b}) \rho}$ and yields 
		\begin{align}
			B(x) \, \bigl ( (-i \partial_{\xi})^a (+i \partial_x)^{\alpha} f \bigr )(X) \, \bigl ( (-i \partial_{\xi})^b (+i \partial_x)^{\beta} g \bigr )(X) 
			. 
		\end{align}
	\end{cor}
	In the proof of Corollary~\ref{appendix:existenceOscInt:Lemma2} we have used that we could write the integrals as a \emph{product} of two independent integrals. There is, however, a second relevant type of oscillatory integral that cannot be `untangled.' Fortunately, we only need to ensure their existence and not evaluate them explicitly. Again, we will start with a simpler integral over only one phase space variable and then extend the ideas to the full integral in a corollary. 
	\begin{lem}\label{appendix:existenceOscInt:remainder}
		Assume $f \in \Hoerrd{m_1}$, $g \in \Hoerrd{m_2}$, $\rho \in [0,1]$, $\eps \in (0,1]$ and $\tau , \tau' \in [0,1]$. Furthermore, let $G_{\tau'} \in \mathcal{BC}^{\infty} \bigl ( \R^d_x,\mathcal{C}^{\infty}_{\mathrm{pol}}(\R^d_y \times \R^d_z) \bigr )$ be such that for all $c , c' , c'' \in \N_0^d$ 
		\begin{align*}
			\babs{\partial_x^c \partial_y^{c'} \partial_z^{c''} G_{\tau'}(x,y,z)} \leq C_{c c' c''} \bigl ( \sexpval{y} + \sexpval{z} \bigr )^{\sabs{c} + \sabs{c'} + \sabs{c''}} 
		\end{align*}
		holds for some finite constant $C_{c c' c''} > 0$. Then for all $a , \alpha , b , \beta \in \N_0^d$ and $\tau , \tau' \in [0,1]$ 
		\begin{align}
			I_{\tau \tau'}(x,\xi) := \frac{1}{(2 \pi)^{2d}} \int \dd Y \int \dd Z \, e^{i \sigma(X,Y+Z)} \, e^{i \tau \frac{\eps}{2} \sigma(Y,Z)} \, G_{\tau'}(x,y,z) \, y^a \eta^{\alpha} z^b \zeta^{\beta} \, (\Fs f)(Y) \, (\Fs g)(Z) 
			\label{appendix:existenceOscInt:eqn:remainder}
		\end{align}
		exists as an oscillatory integral in $\Hoerrd{m_1 + m_2 - (\sabs{a} + \sabs{b}) \rho}$. The map $(\tau , \tau') \mapsto I_{\tau \tau'}$ is continuous. 
	\end{lem}
	\begin{proof}
		Let us rewrite the integral first, the result will serve as a definition for the oscillatory integral $I_{\tau \tau'}$: 
		\begin{align*}
			I_{\tau \tau'}(x,\xi) &= \frac{1}{(2\pi)^{4d}} \int \dd Y \int \dd \tilde{Y} \int \dd Z \int \dd \tilde{Z} \, \bigl ( (+i \partial_{\tilde{\eta}})^a (-i \partial_{\tilde{y}})^{\alpha} e^{i \sigma(X - \tilde{Y} , Y)} \bigr ) \, \bigl ( (+i \partial_{\tilde{\zeta}})^b (-i \partial_{\tilde{a}})^{\beta} e^{i \sigma(X - \tilde{Z} , Z)} \bigr ) 
			\cdot \\
			&\qquad \qquad \qquad \qquad \cdot 
			e^{i \tau \frac{\eps}{2} \sigma(Y,Z)} \, G_{\tau'}(x,y,z) \, f(\tilde{Y}) \, g(\tilde{Z}) 
			\\
			&= \frac{1}{(2\pi)^{4d}} \int \dd Y \int \dd \tilde{Y} \int \dd Z \int \dd \tilde{Z} \, e^{i \sigma(X - \tilde{Y} , Y)} \, e^{i \sigma(X - \tilde{Z} , Z)} \, e^{i \tau \frac{\eps}{2} \sigma(Y,Z)} 
			\cdot \\
			&\qquad \qquad \qquad \qquad \cdot 
			G_{\tau'}(x,y,z) \, \bigl ( (-i \partial_{\tilde{\eta}})^a (+i \partial_{\tilde{y}})^{\alpha} f \bigr )(\tilde{Y}) \, \bigl ( (-i \partial_{\tilde{\zeta}})^b (+i \partial_{\tilde{z}})^{\beta} g \bigr )(\tilde{Z}) 
		\end{align*}
		By assumption $\partial_x^{\alpha} \partial_{\xi}^a f \in \Hoerrd{m_1 - \abs{a} \rho}$ as well as $\partial_x^{\beta} \partial_{\xi}^b g \in \Hoerrd{m_2 - \abs{b} \rho}$ and we see that it suffices to consider the case $a = b = \alpha = \beta = 0$. In this particular case, we estimate all seminorms: let $n , \nu \in \N_0^d$. Then we have to bound 
		\begin{align*}
			\partial_x^n \partial_{\xi}^{\nu} I_{\tau \tau'}(x,\xi) &= \sum_{\substack{a + b + c = n \\ \alpha + \beta = \nu}} \frac{1}{(2\pi)^{4d}} \int \dd Y \int \dd \tilde{Y} \int \dd Z \int \dd \tilde{Z} \, \bigl ( \partial_{x}^a \partial_{\xi}^{\alpha} e^{i \sigma(X - \tilde{Y} , Y)} \bigr ) \, \bigl ( \partial_{x}^b \partial_{\xi}^{\beta} e^{i \sigma(X - \tilde{Z} , Z)} \bigr ) \, e^{i \tau \frac{\eps}{2} \sigma(Y,Z)} 
			\cdot \\
			&\qquad \qquad \qquad \qquad \cdot 
			\partial_x^c G_{\tau'}(x,y,z) \, f(\tilde{Y}) \, g(\tilde{Z}) 
			\\
			&= \sum_{\substack{a + b + c = n \\ \alpha + \beta = \nu}} \frac{1}{(2\pi)^{4d}} \int \dd Y \int \dd \tilde{Y} \int \dd Z \int \dd \tilde{Z} \, e^{i \sigma(X - \tilde{Y} , Y)} \, e^{i \sigma(X - \tilde{Z} , Z)} \, e^{i \tau \frac{\eps}{2} \sigma(Y,Z)} 
			\cdot \\
			&\qquad \qquad \qquad \qquad \cdot 
			\partial_x^c G_{\tau'}(x,y,z) \, \partial_{\tilde{y}}^a \partial_{\tilde{\eta}}^{\alpha} f(\tilde{Y}) \, \partial_{\tilde{z}}^b \partial_{\tilde{\zeta}}^{\beta} g (\tilde{Z}) 
			\\
			&= \sum_{\substack{a + b + c = n \\ \alpha + \beta = \nu}} \int \dd y \int \dd \eta \int \dd z \int \dd \zeta \, e^{i \eta \cdot y} \, e^{i \zeta \cdot z} \, \partial_x^c G_{\tau'}(x,y,z) 
			\cdot \\
			&\qquad \qquad \qquad \qquad \cdot 
			\partial_x^a \partial_{\xi}^{\alpha} f \bigl ( x - \tfrac{\tau \eps}{2} z , \xi - \eta \bigr ) \, \partial_x^b \partial_{\xi}^{\beta} g \bigl ( x + \tfrac{\tau \eps}{2} y , \xi - \zeta \bigr ) 
			. 
		\end{align*}
		from above by an integrable function. To do that, we insert powers of $\sexpval{y}^{-2}$, $\sexpval{z}^{-2}$, $\sexpval{\eta}^{-2}$ and $\sexpval{\zeta}^{-2}$ via the usual trick, \eg $\sexpval{y}^{-2} (1 - \Delta_{\eta}) e^{i \eta \cdot y} = e^{i \eta \cdot y}$. To simplify notation, we set $L_y := 1 - \Delta_y$; $L_z$, $L_{\eta}$ and $L_{\zeta}$ are defined analogously. Then, we have for any $N_1 , N_2 , K_1 , K_2 \in \N_0$ 
		\begin{align}
			\partial_x^n \partial_{\xi}^{\nu} I_{\tau \tau'}(x,\xi) &= \sum_{\substack{a + b + c = n \\ \alpha + \beta = \nu}} \frac{1}{(2\pi)^{2d}} \int \dd y \int \dd \eta \int \dd z \int \dd \zeta \, \bigl ( \sexpval{y}^{-2N_1} L_{\eta}^{N_1} e^{i \eta \cdot y} \bigr ) \, \bigl ( \sexpval{z}^{-2N_2} L_{\zeta}^{N_2} e^{i \zeta \cdot z} \bigr ) 
			\cdot \notag \\
			&\qquad \qquad \qquad \qquad \cdot 
			\partial_x^c G_{\tau'}(x,y,z) \, \partial_x^a \partial_{\xi}^{\alpha} f \bigl ( x - \tfrac{\tau \eps}{2} z , \xi - \eta \bigr ) \, \partial_x^b \partial_{\xi}^{\beta} g \bigl ( x + \tfrac{\tau \eps}{2} y , \xi - \zeta \bigr ) 
			\notag \\
			&= \sum_{\substack{a + b + c = n \\ \alpha + \beta = \nu \\ \sabs{\alpha'} \leq 2 N_1 , \, \sabs{\beta'} \leq 2 N_2}} \negmedspace \negmedspace \negmedspace \negmedspace C_{\alpha' \beta'} \int \dd y \int \dd \eta \int \dd z \int \dd \zeta \, \bigl ( \sexpval{\eta}^{-2K_1} L_{y}^{N_1} e^{i \eta \cdot y} \bigr ) \, \bigl ( \sexpval{\zeta}^{-2K_2} L_{z}^{N_2} e^{i \zeta \cdot z} \bigr ) 
			\cdot \notag \\
			&\qquad \qquad \qquad \qquad \cdot 
			\sexpval{y}^{-2N_1} \, \sexpval{z}^{-2N_2} \, \partial_x^c G_{\tau'}(x,y,z) 
			\cdot \notag \\
			&\qquad \qquad \qquad \qquad \cdot 
			\partial_x^a \partial_{\xi}^{\alpha + \alpha'} f \bigl ( x - \tfrac{\tau \eps}{2} z , \xi - \eta \bigr ) \, \partial_x^b \partial_{\xi}^{\beta + \beta'} g \bigl ( x + \tfrac{\tau \eps}{2} y , \xi - \zeta \bigr ) 
			\notag \\
			&= \sum_{\substack{a + b + c = n \\ \alpha + \beta = \nu \\ \sabs{\alpha'} \leq 2 N_1 , \, \sabs{\beta'} \leq 2 N_2 \\ \sabs{a'} + \sabs{b'} + \sabs{c'} \leq 2 K_1 \\ \sabs{a''} + \sabs{b''} + \sabs{c''} \leq 2 K_2 }} \negmedspace \negmedspace \negmedspace \negmedspace \negmedspace \negmedspace C^{a b c \alpha \beta}_{\alpha' \beta' a' b' c' a'' b'' c''} (\eps \tau)^{\sabs{a''} + \sabs{b'}} \int \dd y \int \dd \eta \int \dd z \int \dd \zeta \, e^{i \eta \cdot y} \, e^{i \zeta \cdot z} \, \sexpval{y}^{-2N_1} \, \sexpval{z}^{-2N_2} 
			\cdot \notag \\
			&\qquad \qquad \qquad \qquad \cdot 
			\sexpval{\eta}^{-2K_1} \, \sexpval{\zeta}^{-2K_2} \, \varphi_{N_1 a'}(y) \, \varphi_{N_2 b''}(z) \, \partial_x^c \partial_y^{x'} \partial_z^{c''} G_{\tau'}(x,y,z) 
			\cdot \notag \\
			&\qquad \qquad \qquad \qquad \cdot 
			\partial_x^{a + a''} \partial_{\xi}^{\alpha + \alpha'} f \bigl ( x - \tfrac{\tau \eps}{2} z , \xi - \eta \bigr ) \, \partial_x^{b + b'} \partial_{\xi}^{\beta + \beta'} g \bigl ( x + \tfrac{\tau \eps}{2} y , \xi - \zeta \bigr ) 
			\label{appendix:existence:remainder:eqn:typical_int}
		\end{align}
		Here, the \emph{bounded} functions $\varphi_{N a}$ are defined by $\partial_x^a \sexpval{x}^{-2N} =: \sexpval{x}^{-2N} \, \varphi_{N a}(x)$ for all $N \in \N_0$, $a \in \N_0^d$, and the constants appearing in the sum are defined implicitly. We now estimate the absolute value of each of the terms in the integral in order to find $N_1 , N_2 , K_1 , K_2 \in \N_0$ large enough so that the right-hand side of the above consists of a finite sum of integrable functions. Using the assumptions on $G_{\tau'}$ and the standard estimate $\sexpval{\xi - \eta}^m \leq \sexpval{\xi}^m \sexpval{\eta}^{\sabs{m}}$, we can bound the integrand of the right-hand side of~\eqref{appendix:existence:remainder:eqn:typical_int} in absolute value by 
		\begin{align*}
			&\sum_{\substack{a + b + c = n \\ \alpha + \beta = \nu \\ \sabs{\alpha'} \leq 2 N_1 , \, \sabs{\beta'} \leq 2 N_2 \\ \sabs{a'} + \sabs{b'} + \sabs{c'} \leq 2 K_1 \\ \sabs{a''} + \sabs{b''} + \sabs{c''} \leq 2 K_2 }} \tilde{C}^{a b c \alpha \beta}_{\alpha' \beta' a' b' c' a'' b'' c''} \, \sexpval{y}^{-2N_1 + \sabs{c} + \sabs{c'} + \sabs{c''}} \sexpval{z}^{-2N_2 + \sabs{c} + \sabs{c'} + \sabs{c''}} \sexpval{\eta}^{-2K_1} \sexpval{\zeta}^{-2K_2} 
			\cdot \\
			&\qquad \qquad \qquad \qquad \cdot 
			\sexpval{\xi - \eta}^{m_1 - (\sabs{\alpha} + \sabs{\alpha'}) \rho} \, \sexpval{\xi - \zeta}^{m_2 - (\sabs{\beta} + \sabs{\beta'}) \rho}
			\\
			&\qquad \qquad \leq 
			C \expval{\xi}^{m_1 + m_2 - \abs{\nu} \rho} \, \sexpval{y}^{-2 N_1 + \abs{n} + 2 K_1 + 2 K_2} \, \sexpval{z}^{-2 N_2 + \abs{n} + 2 K_1 + 2 K_2} \, \sum_{\alpha + \beta = \nu} \sexpval{\eta}^{-2 K_1 + \sabs{m_1 - \sabs{\alpha} \rho}} \, \sexpval{\zeta}^{-2 K_2 + \sabs{m_2 - \sabs{\beta} \rho}} 
			\\
			&\qquad \qquad \leq 
			\tilde{C} \sexpval{\xi}^{m_1 + m_2 - \sabs{\nu} \rho} \, \sexpval{y}^{-2 N_1 + \sabs{n} + 2 K_1 + 2 K_2} \, \sexpval{z}^{-2 N_2 + \sabs{n} + 2 K_1 + 2 K_2} \, \sexpval{\eta}^{-2 K_1 + \sabs{m_1} + \sabs{\nu} \rho} \, \sexpval{\zeta}^{-2 K_2 + \sabs{m_2} + \sabs{\nu} \rho} 
			. 
		\end{align*}
		Choosing $K_1$ and $K_2$ such that $-2 K_j + \sabs{m_j} + \sabs{\nu} \rho < - d$, $j = 1 , 2$, ensures integrability in $\eta$ and $\zeta$. Now that $K_1$ and $K_2$ are fixed, we choose $N_1$ and $N_2$ such that $- 2 N_j + \abs{n} + 2 K_1 + 2 K_2 < - d$ and the right-hand side of the above is an integrable function in $y$, $\eta$, $z$ and $\zeta$ which dominates the absolute value of~\eqref{appendix:existence:remainder:eqn:typical_int}. Thus, we have shown 
		\begin{align*}
			\babs{\partial_x^n \partial_{\xi}^{\nu} I_{\tau \tau'}(x,\xi)} \leq C_{n \nu} \sexpval{\xi}^{m_1 + m_2 - \sabs{\nu} \rho} 
		\end{align*}
		for all $n , \nu \in \N_0^d$ and hence $I_{\tau \tau'}$ exists in $\Hoerrd{m_1 + m_2}$ if the exponents of $y$, $\eta$, $z$ and $\zeta$ in equation~\eqref{appendix:existenceOscInt:eqn:remainder} all vanish, $a = \alpha = b = \beta = 0$. Similarly, for general $a , \alpha , b , \beta \in \N_0^d$, we conclude $I_{\tau \tau'} \in \Hoerrd{m_1 + m_2 - (\sabs{a} + \sabs{b}) \rho}$. Since the above bounds are uniform in $\tau$ and $\tau'$, the continuity of $(\tau , \tau') \mapsto I_{\tau \tau'}$ in the Fréchet topology of $\Hoerrd{m_1 + m_2 - (\sabs{a} + \sabs{b}) \rho}$ follows from dominated convergence. 
	\end{proof}
	%
	

	\section{Details of calculations in example} 
	\label{appendix:details_example}
	
	For convenience of the reader, we present the derivation of $h_{\mathrm{eff} \, 31}$ and $h_{\mathrm{eff} \, 33}$ in more detail. The calculation simplifies tremendously when one discards blockoffdiagonal terms as soon as possible. To that effect, we use that $\alpha_j$ and $\beta$ \emph{anti}commute, $\beta^2 = \id_{\C^4}$ and $(\xi \cdot \alpha) \, (\xi \cdot \alpha) = \xi^2 \, \id_{\C^4}$. Furthermore, we introduce the spin operators 
	\begin{align*}
		\rho_j = \left (
		\begin{matrix}
			\sigma_j & 0 \\
			0 & \sigma_j \\
		\end{matrix}
		\right )
		, 
		&& 
		j = 1 , 2 , 3 
		. 
	\end{align*}
	Then, we plug in the definition of $h_{\mathrm{eff} \, 31}$ and keep only terms that are purely blockdiagonal, 
	\begin{align*}
		h_{\mathrm{eff} \, 31} 
		&= - i \partial_{x_l} V \, \frac{1}{\sqrt{2E(E+m)}} \, \piref \, \left [ 
			- \frac{m \xi_l}{2 \sqrt{2} E^{\nicefrac{5}{2}} (E+m)^{\nicefrac{1}{2}}} \, \id_{\C^4} \, (E+m) \, \id_{\C^4}
			 + \right . \\
			&\qquad \qquad \qquad \left . 
			+ \frac{\xi_l (2E+m)}{2 \sqrt{2} E^{\nicefrac{5}{2}} (E+m)^{\nicefrac{3}{2}}} (\xi \cdot \alpha) \beta \, (\xi \cdot \alpha) \beta 
			- \frac{1}{\sqrt{2 E (E+m)}} \alpha_l \beta \, (\xi \cdot \alpha) \beta 
		\right ] \, \piref 
		. 
	\end{align*}
	Writing out $E = \sqrt{m^2 + \xi^2}$ and using $(\nabla_x V \cdot \alpha) \, (\xi \cdot \alpha) = (\nabla_x V \cdot \xi) \, \id_{\C^4} + i (\nabla_x V \wedge \xi) \cdot \rho$, we get 
	\begin{align*}
		h_{\mathrm{eff} \, 31} 
		&= - i \, \partial_{x_l} V \, \piref \, \left [ 
			- \frac{m \xi_l}{4 E^3} \, \id_{\C^4} 
			- \frac{\xi_l (2E+m)}{4 E^3 (E+m)^2} \xi^2 \, \id_{\C^4} 
			+ \frac{1}{2 E(E+m)} \alpha_l (\xi \cdot \alpha) 
		\right ] \, \piref \\
		%
		%
		&= \frac{i \, (\nabla_x V \cdot \xi)}{4 E^3 (E+m)^2} \; 
			\bigl ( 
				2 m E^2 + 2 E (m^2 + \xi^2) - 2 E^2 (E+m)
			\bigr ) \, \piref 
			- \frac{i^2}{2 E(E+m)} \piref \, (\nabla_x V \wedge \xi) \cdot \rho \, \piref 
		\\
		&= + \frac{1}{2 E(E+m)} (\nabla_x V \wedge \xi) \cdot \sigma 
		. 
	\end{align*}
	Similarly, we can compute $h_{\mathrm{eff} \, 33}$, 
	\begin{align*}
		h_{\mathrm{eff} \, 33} &= \piref \, \bigl ( 
			(u_0 \magBc H_0)_{(3)} - (h_0 \magBc u_0)_{(3)}
		\bigr ) \, {u_0}^{\ast} \, \piref 
		= \tfrac{i}{2} B_{lj} \, \piref \, \bigl ( 
			\partial_{\xi_l} u_0 \, \partial_{\xi_j} H_0 - \partial_{\xi_l} h_0 \, \partial_{\xi_j} u_0 
		\bigr ) \, {u_0}^{\ast} \, \piref
		\\
		&= \frac{i}{2} B_{lj} \, \piref \, \left [ 
			\left ( 
				- \frac{m \xi_l}{2 \sqrt{2} E^{\nicefrac{5}{2}} (E+m)^{\nicefrac{1}{2}}} \, \id_{\C^4} 
				+ \frac{\xi_l (2E+m)}{2 \sqrt{2} E^{\nicefrac{5}{2}} (E+m)^{\nicefrac{3}{2}}} (\xi \cdot \alpha) \beta 
				- \frac{1}{\sqrt{2 E (E+m)}} \alpha_l \beta 
			\right ) \, \alpha_j 
		+ \right . \\
		& \left . \quad
			- \frac{\xi_l}{E} \beta \, \left ( 
				- \frac{m \xi_j}{2 \sqrt{2} E^{\nicefrac{5}{2}} (E+m)^{\nicefrac{1}{2}}} \, \id_{\C^4} 
				+ \frac{\xi_j (2E+m)}{2 \sqrt{2} E^{\nicefrac{5}{2}} (E+m)^{\nicefrac{3}{2}}} (\xi \cdot \alpha) \beta 
				- \frac{1}{\sqrt{2 E (E+m)}} \alpha_j \beta 
			\right ) 
		\right ] \, {u_0}^{\ast} \, \piref 
		. 
	\end{align*}
	Once we use $B_{lj} \, \xi_l \xi_j = 0$ and plug ${u_0}^{\ast}$ back in, we get the claim, 
	\begin{align*}
		h_{\mathrm{eff} \, 33} 
		&= \frac{i}{2} B_{lj} \, \piref \, \left [ 
				- \frac{m \xi_l}{2 \sqrt{2} E^{\nicefrac{5}{2}} (E+m)^{\nicefrac{1}{2}}} \, \alpha_j 
				+ \frac{\xi_l (2E+m)}{2 \sqrt{2} E^{\nicefrac{5}{2}} (E+m)^{\nicefrac{3}{2}}} (\xi \cdot \alpha) \beta \, \alpha_j 
			+ \right . \\
			&\qquad \qquad \qquad \qquad \qquad \qquad \left . 
				- \frac{1}{\sqrt{2 E (E+m)}} \alpha_l \beta \, \alpha_j 
				+ \frac{\xi_l}{\sqrt{2} E^{\nicefrac{3}{2}} (E+m)^{\nicefrac{1}{2}}} \beta \alpha_j \beta 
		\right ] \, {u_0}^{\ast} \, \piref 
		\\
		&= \frac{i \, B_{lj}}{8 E^3 (E+m)} \, \piref \Bigl ( 
			- \xi_l (2E+m) \, \alpha_j \, (\xi \cdot \alpha) 
			- \xi_l (2E+m) (\xi \cdot \alpha) \alpha_j 
			+ 2 E^2 (E+m) \alpha_l \alpha_j 
		\Bigr ) \, \beta \, \piref 
		\\
		&= - \frac{i \, 2 B_{lj} \, \xi_l \, \xi_j \, (2E+m)}{8 E^3 (E+m)} \, 
		\piref 
		- \frac{1}{2 E} \piref \, B \cdot \rho \, \piref
		= - \frac{1}{2 E} B \cdot \sigma 
		. 
	\end{align*}
	%
	
\end{appendix}

\bibliographystyle{alpha}
\bibliography{quantization}

\end{document}